\documentclass{article}
\usepackage{ijcai24}

\usepackage{times}
\usepackage{soul}
\usepackage{url}
\usepackage[hidelinks]{hyperref}
\usepackage[utf8]{inputenc}
\usepackage[small]{caption}
\usepackage{graphicx}
\usepackage{amsmath}
\usepackage{amsthm}
\usepackage{booktabs}
\usepackage{algorithm}
\usepackage{algorithmic}
\usepackage[switch]{lineno}
\usepackage{bbm}
\usepackage{natbib}
\usepackage{amssymb}
\usepackage{xcolor}
\usepackage{tikz}
\usepackage{enumerate}
\newcommand\overlinex{}
\usepackage{mathtools}
\DeclareMathOperator*{\argmax}{arg\,max}

\newtheorem{theorem}{Theorem}
\newtheorem{lemma}{Lemma}
\newtheorem{corollary}{Corollary}

\title{Multiplayer General Lotto Blotto game}

\author{
Yan Liu$^1$
\and
Bonan Ni$^2$
\and
Weiran Shen$^{1}$
\and
Zihe Wang$^1$
\and
Jie Zhang$^3$
\affiliations
$^1$Renmin University of China \\
$^2$Tsinghua University\\
$^3$University of Bath\\
\emails
liuyan5816@ruc.edu.cn,
bricksern@gmail.com,
shenweiran@ruc.edu.cn,
wang.zihe@ruc.edu.cn,
jz2558@bath.ac.uk
}

\begin{document}

\maketitle

\begin{abstract}
In this paper, we investigate the multiplayer General Lotto game -- a significant extension of the classic Colonel Blotto game -- in a setting involving competition across multiple battlefields.
Under this framework, resources are allocated probabilistically by players, ensuring their expected expenditures remain within individual budgets.
Our contributions begin by establishing the existence of Nash equilibrium for general scenarios, accommodating asymmetries in both player budgets and battlefield valuations.
A detailed characterization of equilibrium strategies follows, specifically addressing the complexities arising when multiple players compete on a single battlefield, and culminating in a system of equations for computing equilibria.
Furthermore, we identify conditions under which equilibrium uniqueness is guaranteed in single-battlefield game.
Turning to multi-battlefield competition, our analysis reveals an upper bound on the average number of battlefields actively contested by each player.
For symmetric scenarios, we provide explicit equilibrium solutions.
Finally, equilibrium multiplicity is demonstrated concretely through an illustrative example involving multiple players and battlefields.
\end{abstract}

\section{Introduction}
The Colonel Blotto game (CB game) is one of the simplest and most well-known game-theoretic models for resource allocation.
Initially proposed by \citep{r1}, the CB game has been widely studied over the years \citep{r4}.
In this game, two competitors, $A$ and $B$, are each given a fixed budget of resources to compete over $n$ battlefields. 
Each player assigns a specific value to each battlefield and simultaneously allocates their resources across these $n$ battlefields, ensuring that the total allocation does not exceed their respective budgets. 
On each battlefield, the player with the highest allocation wins and receives the value assigned to that battlefield, while the losing player receives no reward.
Consequently, each player's utility is the sum of their gains from all $n$ battlefields.
The objective for each player is to maximize their own utility by strategically choosing their resource distribution.

In addition to the standard CB game, several variants have been proposed \citep{r25,r20}.
One of the most well-known is the General Lotto (GL) game, where the budget constraint is relaxed so that aggregate allocations do not exceed a player's budget only in expectation, rather than with certainty \citep{r12,r13,r14,r16,r17,beale-1962,sahuguet-2006}.
The GL game models scenarios in which multiple items are allocated through repeated, independent all-pay auctions, and several budget-constrained bidders compete for these items on a daily basis.
For example, consider the case where several machine learning competitions start online every day, each with a new task.
Multiple entities, such as academic research groups or companies, compete for the first place in each competition.
These entities may value each competition's first place differently.
Each entity has a fixed amount of computing power, determining the total computation they can perform each day.
For every competition, the entity that spends the most computation on its task wins.
To increase the probability of winning a competition, an entity must allocate more computation to it, at the cost of either (i) having less computation to spend on other competitions, or (ii) performing poorly in competitions that start in the following days (assuming each competition has a relatively long duration, allowing entities to accumulate several days' worth of computation for a single competition).
In fact, as demonstrated in the proof of Theorem \ref{thm-mbeq-exist}, there exists a threshold amount of computation beyond which no player finds it beneficial to accumulate more for any single competition.

To our knowledge, in non-cooperative game environment, current research on the GL game with asymmetric budgets has been limited to the two-player case \citep{r12,r13,r14,r16,Vu-2021}.
However, many real-world scenarios involve multiple participants, such as market competition, international relations, social networks, and ecosystems.
These situations all involve strategic resource allocation and competition, and a GL game framework with multiple players can more accurately model interactions.
Additionally, in applications such as auction design and public policy making, designers aim to guide participants toward a specific equilibrium by adjusting rules or mechanisms.
Research on Nash equilibrium in multiplayer settings can aid in identifying and designing effective mechanisms for these purposes.

Investigating the existence, uniqueness, and structure of Nash equilibrium in multiplayer settings is an important but highly challenging task.
In the context of two players, \cite{r14} completely characterized the Nash equilibria in a general setting in which players have asymmetric budgets and the battlefield values are also asymmetric between players.
In this paper, we extend their setting from two players to multiple players.
However, analyzing the Nash equilibrium in multiplayer settings introduces certain technical challenges.
These challenges mainly fall into the following three aspects:
\begin{enumerate}
    \item In the case of a single battlefield, \cite{r14} showed that when the GL game involves only two players, there exists a closed-form solution for the Nash equilibrium.
    When there are only two players, their supports are identical in Nash equilibrium.
    However, when the number of players in the GL game exceeds two, such a closed-form solution no longer exists.
    We discover that the upper endpoints of the supports of players' strategies coincide, and that the minimum value of a player's support above zero inversely correlates with his budget in Nash equilibrium.
    \item In the case of multiple battlefields, \cite{r14} proved that the GL game has a unique set of Nash equilibrium univariate marginals.
    They used a parameter related to players' budgets to partition the battlefields into regions where each player has an advantage.
    However, in the multi-player, multi-battlefield setting, the set of Nash equilibrium univariate marginals is not unique.
    The multiplicity of Nash equilibria makes analyzing Nash equilibria more challenging.
    \item Regarding the existence of Nash equilibrium, \cite{r14} employed a constructive proof approach.
    However, when multiple players participate in the GL game, this constructive method becomes inapplicable due to the complicated computation of Nash equilibrium.
    We use a ``discretization + limiting'' game framework, applying Kakutani's fixed point theorem and Helly's selection theorem to prove the existence of Nash equilibria.
\end{enumerate}

\subsection{Our Contribution}
In this work, we focus on the multiplayer GL game with asymmetric budgets over multiple heterogeneous battlefields.
Our key contributions are as follows:
\begin{itemize}
    \item We establish the existence of a Nash equilibrium in the GL game.
    We start with constructing a variant of the GL game where each player's bid space is discrete and bounded from above, and we demonstrate the existence of a Nash equilibrium in this modified game.
    Subsequently, we show that if the threshold is sufficiently large, it becomes non-restrictive.
    A Nash equilibrium in the GL game arises from the limit of a sequence of Nash equilibria in the modified games, where the bid grid in the sequence becomes finer and finer.
    A more detailed discussion of the approach is given at the end of Section 3.
    \item For the game with a single battlefield, we provide a comprehensive characterization of Nash equilibrium, revealing a relationship between the relative order of players' budgets and the support of their strategies.
    This characterization also naturally implies the known result for the Nash equilibrium in the case of two players with asymmetric budgets \citep{r14}.
    Additionally, we provide a system of equations to solve for the Nash equilibria.
    Moreover, we prove the uniqueness of Nash equilibrium when there are at least two players with the maximum budget.
    \item For the game with multiple battlefields, we prove that for almost every value profile, each player focuses only on few battlefields when the number of players is sufficiently large.
    We show that Nash equilibrium is not unique by providing an example.
    Additionally, in the symmetric setting with multiple players and multiple battlefields, we present a solution for the Nash equilibrium.
\end{itemize}
To our knowledge, this is the first paper to study the multiplayer GL game with asymmetric budgets.
Our results significantly extend the existing literature on the GL game.


\subsection{Related Work}
The CB game was initially developed to simulate military logistics, where resources are analogous to soldiers, equipment, or weapons.
Owing to its ability to model a variety of real-world scenarios, the CB game has aroused considerable interest among scholars in fields such as sociology, mathematics, economics, and computer science \citep{r5,r6,r7,r8,r9,r11,r10,Vu-2021,liu-Stackelberg}.
In the CB game, there may not be pure strategy Nash equilibrium; however a mixed strategy Nash equilibrium does exist, represented by a pair of $n$-variate distributions \citep{r11,Perchet-2022,Macdonell-2015,liu-2025}.
Research on the CB game primarily focuses on identifying and computing Nash equilibria \citep{r11,r21,r23,r24}.
However, the computation of Nash equilibrium in this game is nontrivial, because the equilibrium strategies correspond to complicated joint distributions defined on an $n$-dimensional simplex \citep{r11,r19,r18,r20,r21}.

The GL game, as the most well-known variant of the CB game, has also been extensively studied.
\cite{beale-1962} were the first to introduce the GL game, proposing it as an auxiliary construct for deriving approximate solutions to the Colonel Blotto game.
\cite{r16} considered the GL game under an information asymmetry setting, where one player's budget was public knowledge while the other player's budget was drawn from a Bernoulli distribution.
\cite{r17} studied the GL game with a concession, where the players' concession strategies enables players to reach a more effective Nash equilibrium in a competitive environment.
\cite{Chandan-2020} investigated how, in the GL game, openly declaring strategic intentions can influence an opponent's strategy choices, helping players to secure better outcomes in specific situations and gain an advantage.
In addition, \cite{Chandan-2022} explored the optimal strategies for resource allocation in a multi-stage GL game.
Their work indicated that by dynamically adjusting investment strategies, players can significantly increase their winning rates at different stages, thereby maximizing overall returns.

\section{Preliminaries} 
In the multiplayer General Lotto (GL) game, a set of $n$ players, indexed by $[n] := \{1, 2, \cdots, n\}$, compete across $m$ battlefields, indexed by $[m] := \{1, 2, \cdots, m\}$. 
Each player $i \in [n]$ is equipped with a fixed budget $B_i > 0$ and a valuation $v_{i,j} > 0$ for each battlefield $j \in [m]$. 
A player's \emph{pure strategy} is to distribute their budget, referred to as \emph{bids}, across the battlefields.
The outcome on each battlefield is determined independently; the player who places the highest bid on a battlefield wins that battlefield and gains the associated value.
All other players receive no reward for that battlefield.
In cases where multiple players tie for the highest bid, the winner is chosen uniformly at random.

Following the model introduced by \citep{r14}, we consider mixed strategies in which a player's bid on each battlefield is modeled as a nonnegative random variable.
That is, instead of choosing fixed bids, player $i$ selects a tuple $(F_{i,j})_{j \in [m]}$, where each $F_{i,j}$ is a cumulative distribution function (c.d.f.) from which their actual bid, denoted $X_{i,j}$, is drawn.
The total expected bid across all battlefields must respect the budget constraint. 
The set of feasible strategies $\mathcal{F}_i$ includes all such $m$-tuples $(F_{i,j})_{j \in [m]}$ that meet the following conditions: \\
$(i)$ the support of each $F_{i,j}$, denoted by $Supp_{i,j}$, is contained in $[0, +\infty)$, ensuring all bids are nonnegative, and \\
$(ii)$ the total expected bid across all battlefields does not exceed player $i$'s budget $B_i$. \\
We therefore represent the set of feasible strategies for player $i$ as
\begin{align}\label{eq-strategy}
    \mathcal{F}_i \mathrel{\mathop:}= \Big\{ (F_{i,j})_{j \in [m]}: \,\,\ \sum\nolimits_{j=1}^{m} \mathbb{E}_{X_{i,j} \sim F_{i,j}} \left[ X_{i,j} \right] \leq B_i, \,\,\ Supp_{i,j} \subseteq [0, +\infty) \Big\}.
\end{align}

We use $F_{i,j}(x^+)$ and $F_{i,j}(x^-)$ to denote the right-hand and left-hand limits of the distribution function $F_{i,j}$ at point $x$, respectively.
Throughout the paper, the notation $-i$ refers to all players except player $i$. 
On each battlefield $j$, the bids $(X_{i,j})_{i \in [n]}$ are independently distributed.
Specifically, suppose the realization of the random variable $X_{i,j}$ takes the value $x \ge 0$, while the bids of all other players on battlefield $j$, $(X_{{i'},j})_{i'\neq i}$, are drawn independently from their respective distributions $F_{-i,j}$.
Then, the probability that player $i$ wins battlefield $j$ is given by
\begin{align}\label{eq-utility}
    \text{Pr}[i \text{ wins } j \text{ by bidding } x] = \mathop{\mathbb{E}}\limits_{X_{-i,j} \sim F_{-i,j}}\left[\frac{\mathbb{I}[x \ge X_{i',j}, \forall i' \ne i]}{\# \{i' \ne i : X_{i',j} = x\} + 1}\right],
\end{align}
where $\mathbb{I}[\cdot]$ is the indicator function, equal to 1 if the predicate is true and 0 otherwise, and $\# \{ \cdot \}$ denotes the cardinality of a set, i.e., the number of elements in the set.

Player $i$'s expected utility on battlefield $j$ is given by their valuation $v_{ij}$ multiplied by the probability of winning that battlefield.
Their total expected utility is the sum of expected utilities across all battlefields.
These are formally defined as:
\begin{align*}
    & u_{i,j}(x, F_{-i, j}) = v_{ij} \cdot  \underset{X_{-i, j} \sim F_{-i,j}}{\mathbb{E}} \left[\frac{\mathbb{I}[x \ge \max_{i' \in [n]} X_{i', j}]}{\# \{i' \ne i: X_{i',j} \ge x\} + 1}\right]. \\
    & U_{i} \left( (F_{i,j})_{j \in [m]}, (F_{-i, j})_{j \in [m]} \right) = \sum\nolimits_{j \in [m]} \underset{X_{i,j} \sim F_{i,j}}{\mathbb{E}} \left[u_{i,j}(X_{i,j}, F_{-i,j}) \right].
\end{align*}

Players are utility maximizers. Given the strategy profile of the other players
$(F_{-i, j})_{j \in [m]}$, player $i$'s best response is a strategy that maximizes their expected utility:
\begin{align*}
    (F_{i,j})_{j \in [m]} \in \argmax_{(F'_{i,j})_{j \in [m]} \in \mathcal{F}_i} U_i \left( (F'_{i,j})_{j \in [m]}, (F_{-i, j})_{j \in [m]} \right).
\end{align*}
A strategy profile $(F_{i, j})_{i \in [n], j \in [m]}$ is a Nash equilibrium if, for every player $i$, their strategy is a best response to the strategies of all other players.

Throughout the paper, we focus on the mixed strategies of the players, i.e., $F_{i,j}$, $\forall i \in [n]$, $\forall j \in [m]$.

The following lemma provides a necessary condition for a strategy profile to be a Nash equilibrium.
Intuitively, on each battlefield, a player's utility must be bounded above by a linear function of their own bid.
Moreover, the player assigns positive probability only to those bids where this upper bound is tight.
That is, the utility equals the value of the linear function.
Otherwise, the player could deviate to a different bid that yields strictly higher utility without altering the expected allocation to that battlefield.
Additionally, for each player, the slopes of these linear upper bounds must be the same across all battlefields.
If they differed, the player could improve their utility by shifting expected allocations between battlefields.
\begin{lemma}\label{lem-cont-linear}
    If $(F_{i, j})_{i \in [n], j \in [m]}$ is a Nash equilibrium, then there exist constants $a_i > 0$ for every $i \in [n]$ and $b_{i,j} \ge 0$ for every $i \in [n], j \in [m]$ such that
    \begin{align*}
        \Pr_{X_{i,j} \sim F_{i, j}} \left[u_{i,j}(X_{i,j}, F_{-i, j}) = a_{i} X_{i,j} + b_{i,j} \right] = 1,
    \end{align*}
    and for all $x \geq 0$, it holds that
    \begin{align*}
        u_{i,j}(x, F_{-i, j}) \le a_{i} x + b_{i,j},
    \end{align*}
    for every $i \in [n]$ and $j \in [m]$.
\end{lemma}

We next establish a property of the constants $b_{i,j}$ that further simplifies the analysis.
\begin{lemma}\label{lem-one-pos-b}
   For every battlefield $j \in [m]$, at most one player can have $b_{i,j} > 0$. That is,   
   $\# \{i:b_{i,j} > 0\} \le 1$.
\end{lemma}
This result will be particularly useful in our later analysis, especially when characterizing Nash equilibria on a single battlefield.

\section{Existence of Nash Equilibrium in the General Lotto game}
In this section we show that the GL game has a Nash equilibrium.

\begin{theorem}\label{thm-mbeq-exist}
   A Nash equilibrium exists in the GL game.
\end{theorem}

The proof of Theorem \ref{thm-mbeq-exist} relies on the following game $\mathcal{G}_k$, which is a modification of the GL game where only bounded bids on discrete grids are allowed, and ties are broken uniformly at random.

For every $k = 1, 2, \ldots$, consider the following game $\mathcal{G}_k$:
\begin{itemize}
    \item The finite set of feasible bids is given by $A_k \mathrel{\mathop:}= \{\frac{l}{k} \cdot T \mid 0 \le l \le k, l \in \mathbb{Z} \}$, where $T = 2^{2n+3} \cdot \max_{i \in [n]} B_i$.
    \item Every player $i \in [n]$ chooses a strategy, which is given by $m$ distributions $F_{i,\cdot,k} \mathrel{\mathop:}= (F_{i, j, k})_{j \in [m]}$ with every $F_{i,j,k} \in \Delta(A_k)$, so that player $i$'s random bid on battlefield $j$ is given by $X_{i,j} \sim F_{i,j,k}$, satisfying constraint $\sum_{j \in [m]}\mathbb{E}[X_{i,j}] \le B_i$.
    Note that the bids $(X_{i,j})_{i \in [n], j \in [m]}$ are independently distributed.
    \item Denote the strategies of all players other than $i$ on battlefield $j$ by $F_{-i, j, k} \mathrel{\mathop:}= (F_{i',j,k})_{i' \ne i}$, and the random bids of all players other than $i$ on battlefield $j$ by $X_{-i, j} \mathrel{\mathop:}= (X_{i', j})_{i' \ne i}$.
    Denote the strategy profile by $F_{\cdot,\cdot,k} \mathrel{\mathop:}= (F_{i,\cdot,k})_{i \in [n]}$, and the strategies of players other than $i$ by $F_{-i,\cdot, k}$.
    \item Denote the set of $i$'s feasible strategies by $\mathcal{F}_{i,k}$, which is a compact subset of $(\Delta(A_k))^m$.
    Denote the set of feasible strategy profiles by $\mathcal{F}_k \mathrel{\mathop:}= \times_{i \in [n]} \mathcal{F}_{i,k}$.
    \item Break ties uniformly at random: For each $i \in [n]$ and each $F_{\cdot, \cdot, k} \in \mathcal{F}_k$, the utility of player $i$'s is given by the function $U_i: \mathcal{F}_k \rightarrow \mathbb{R}$, defined as:
    \begin{align*}
        u_{i,j}(x, F_{-i, j, k}) = \underset{X_{-i, j} \sim F_{-i,j,k}}{\mathbb{E}} \left[\frac{ v_{ij} \cdot \mathbb{I}[x \ge \max_{i' \in [n]} X_{i', j}]}{ \# \{i' \ne i: X_{i',j} \ge x \} + 1} \right].
    \end{align*}
    \begin{align*}
        U_{i}(F_{\cdot,\cdot,k}) = \sum_{j \in [m]} \mathbb{E}_{X_{i,j} \sim F_{i,j,k}} \left[u_{i,j}(X_{i,j}, F_{-i,j,k}) \right].
    \end{align*}
\end{itemize}
Throughout this section, given a positive integer $k$, we denote a Nash equilibrium of $\mathcal{G}_k$ by $\tilde F_{\cdot,\cdot,k}$ and represent a Nash equilibrium strategy of player $i$ by $(\tilde F_{i,\cdot,k})_{j \in [m]}$.

Based on our constructed discrete game $\mathcal{G}_k$, we adopt the following steps to prove the existence of Nash equilibrium in the GL game:
\begin{itemize}
    \item Show that $\mathcal{G}_k$ admits a Nash equilibrium.
    \item Show that no player bids in $[\frac{T}{2} + \frac{T}{k}, T]$ when $k$ is sufficiently large.
    \item Show that $\tilde F_{i,j,k}$ converges when $k \to \infty$.
    \item Show that $\tilde F_{i,j,k}$ is uniformly continuous in the interval $[0,T]$ when $k \to \infty$.
    \item Show that the existence of a Nash equilibrium in $\mathcal{G}_k$ can be extended to the original game GL.
\end{itemize}

One can apply Kakutani's fixed point theorem to a proper set-valued function $\beta: \mathcal{F}_k \rightarrow 2^{\mathcal{F}_k}$ and conclude that a Nash equilibrium exists in $\mathcal{G}_k$.

\begin{lemma}\label{lem-gk-nash}
    For every $k \ge 1$, $\mathcal{G}_k$ has a Nash equilibrium.
\end{lemma}

It should be noted that Nash's theorem cannot be used to prove Lemma \ref{lem-gk-nash}.
Under Nash's theorem, each player's feasible mixed strategy set is a simplex, with each vertex corresponding to a pure strategy, and every pure strategy for a player must not exceed their budget.
However, in our model, each player's strategy must satisfy the constraint that the expected total allocation does not exceed their budget.
This means that, on battlefield $j$, player $i$'s mixed strategy may assign positive probability to bids that individually exceed their budget.

For any $\tilde F_{\cdot, \cdot, k}$, Lemma \ref{lem-bounded-eq} shows that the supports of the Nash equilibrium strategies are uniformly bounded away from the upper threshold $T$.
This means that, as long as $T$ is chosen sufficiently large relative to the players' budgets, the actual Nash equilibrium strategies are unaffected by the specific value of $T$.
The intuition is as follows: if a player bids a sufficiently high value, such as $3 \cdot \max_{i\in [n]} B_i$, then by Markov's inequality and the expected budget constraints, the probability that any other player's bid exceeds this amount must be less than $\frac{1}{2}$.
Thus, a player who bids $3 \cdot \max_{i \in [n]} B_i$ can beat any single opponent with probability greater than $\frac{1}{2}$, and wins against all opponents simultaneously with probability greater than $(\frac{1}{2})^{n-1}$.
However, such a high bid is very costly in expectation and not optimal under the budget constraint.
Therefore, in Nash equilibrium, players will not put probability mass near the upper threshold, and the strategies concentrate on much lower bids.

\begin{lemma} \label{lem-bounded-eq}
    Given $T = 2^{2n+3} \cdot \max_{i \in [n]} B_i$, there exists $K > 0$ s.t. for any $k \ge K$, no player bids in $[\frac{T}{2} + \frac{T}{k}, T]$ on any battlefield $j$ in Nash equilibrium of $\mathcal{G}_k$.
\end{lemma}

Next, we show that the sequence $\{F_{i,j,k}\}_{k\ge 1}$ admits a subsequence that converges to a limit distribution when $k \to \infty$.
Because every distribution function $F_{i,j,k}$ is monotone and bounded in $[0,1]$, the family $\{F_{i,j,k}\}_{k \ge 1}$ forms a set of bounded monotone functions.
Helly's selection theorem guarantees that for any fixed coordinate $(i,j)$, one can extract a pointwise-convergent subsequence.
Since the numbers of players $n$ and battlefields $m$ are finite, we iterate over the coordinates $(1,1),(1,2),\dots,(n,m)$: first take a convergent subsequence for $(1,1)$; within that subsequence, extract another for $(1,2)$; and so on.
After finitely many steps we obtain an infinite subsequence $\{k_\ell\}$ with $k_\ell \to \infty$ along which $F_{i,j,k_\ell}(x)$ converges for every $x$ and for every $(i,j)$.
Thus, the limits $F_{i,j}(x) = \lim_{\ell \to \infty} F_{i,j,k_\ell}(x)$ define the desired limit distributions.
Hence, using Helly's ``coordinate-by-coordinate, nested subsequence'' procedure and the fact that the total number of coordinates is finite, we ensure that the full strategy sequence converges simultaneously along a common subsequence.

In the remaining part of this section, we refine $(\tilde F_{\cdot,\cdot,k})_{k \ge 1}$ to be the final subsequence obtained after the above steps, in which every $(\tilde F_{i,j,k})_{k \ge 1}$ converges pointwise.
Denote the final limit of the sequence as $\tilde F_{\cdot, \cdot} \mathrel{\mathop:}= (\tilde F_{i,j})_{i \in [n], j \in [m]}$.
Note that for $\tilde F_{i,\cdot,k}$ to be a best response to $\tilde F_{-i,\cdot,k}$, it must satisfy the condition $\sum_{j \in [m]} \mathbb{E}_{X_{i,j} \sim \tilde F_{i,j,k}} [X_{i,j}] = B_i$, which in turn implies that $\sum_{j \in [m]} \mathbb{E}_{X_{i,j} \sim \tilde F_{i,j}} [X_{i,j}]$ $= B_i$.

Note that although every c.d.f. $\tilde F_{i,j,k}$ is right-continuous, the limit $\tilde F_{i,j}$ does not necessarily have continuity.
By taking some sufficiently large $k$ and analyzing Nash equilibrium of $\mathcal{G}_k$, we can establish some properties about the continuity of $\tilde F_{\cdot, \cdot}$ in $[0, T)$ (see Lemma \ref{lem-zero-cont} and Lemma \ref{lem-continuous} in Appendix \ref{appendix-sec-3}).
Based on the continuity, we can derive Lemma \ref{lem-uniform-converge}, which establishes that the convergence of $(\tilde F_{\cdot, \cdot, k})_{k \ge 1}$ is uniform and furthermore ensures that the sequences $u_{i, j}$ also converge.

\begin{lemma}\label{lem-uniform-converge}
    Given $n \ge 3$, for every $i \in [n], j \in [m]$, we have:
    \begin{enumerate}
        \item $\tilde F_{i,j}$ is uniformly continuous on $[0,T]$,
        \item the sequence $(\tilde F_{i,j,k})_{k \ge 1}$ uniformly converges to $\tilde F_{i,j}$,
        \item the sequence of univariate functions $u_{i, j}(\cdot, \tilde F_{-i,j,k})$ will converge in $||\cdot||_{\infty}$ to $u_{i,j}(\cdot, \tilde F_{-i, j})$, where $u_{i,j}(\cdot, \tilde F_{-i, j})$ is the utility function of original GL game with continuous bid space.
        That is,
        \begin{align*}
            \lim_{k \to \infty} \sup_{x \in [0,T]}|u_{i, j}(x, \tilde F_{-i,j,k}) - u_{i,j}(x, \tilde F_{-i, j})| = 0.
        \end{align*}
    \end{enumerate}
\end{lemma}

Lemma \ref{lem-uniform-converge} leads to the conclusion that $\tilde F_{\cdot, \cdot}$ is a Nash equilibrium in the game where the players can choose random bids within interval $[0, T]$.
By Lemma \ref{lem-bounded-eq}, we can know that no player bids in $[\frac{3T}{4}, T]$.
In fact, even if players are allowed to bid more than $T$, in Nash equilibrium, no player will actually bid above $T$.
Thus, we can establish the existence of a Nash equilibrium in the GL game.
Due to space limitations, the proof of Theorem \ref{thm-mbeq-exist} is deferred to Appendix \ref{appendix-sec-3}.

\paragraph{Remark on the proof approach.}
In the General Lotto game, the existence of a Nash equilibrium cannot be established directly by a fixed-point theorem.
In particular, even a more powerful variant of the Kakutani theorem, namely the Fan--Glicksberg theorem, is not applicable.
The reasons are as follows:
To prove the existence result using a fixed-point theorem, we need to consider players' best responses.
Let the players' best response be represented by a correspondence $\Gamma: F \rightrightarrows F$.
The correspondence $\Gamma_i: F_{-i} \rightrightarrows F_i$ is required to satisfy the following properties:
first, $\Gamma_i$ must be compact, convex, and nonempty;
second, $\Gamma_i$ must be upper hemicontinuous (u.h.c.).
However, in our game, neither of these properties holds naturally.
For the first property, because of the tie-breaking rule, a best response may fail to exist.
For instance, assuming there are only two players competing, if player 1 bids 1 with certainty, then there is no best response for player 2.
In addition, since the action space is unbounded, the value of $\Gamma_i$ may also be unbounded.
For the second property, because of the tie-breaking rule, $\Gamma_i(\lim\limits_{F_{-i} \rightarrow F_{-i}^*} F_{-i})$ may fail to exist, where $F_{-i}^*$ denotes the limit of the sequence $F_{-i}$.
Even if it exists, the utility function is not continuous in the presence of ties; hence the best response mapping fails to be upper hemicontinuous.
These obstacles make a direct application of the fixed-point theorem impossible, which explains the technical necessity of our discretization-based approach.


\section{Nash Equilibrium Characterization of A Single Battlefield}
In this section, we analyze Nash equilibrium in the case of a single battlefield, i.e., the special case where $m = 1$.
We begin by providing a complete characterization of the structure of Nash equilibrium.
Our analysis reveals that the upper endpoints of the supports of the players' Nash equilibrium strategies coincide, and the minimum value of a player's support above zero is inversely correlated with their budget.
Furthermore, we prove the uniqueness of the Nash equilibrium when at least two players share the maximum budget.

\subsection{Characterization of Nash Equilibrium}
With only one battlefield, we omit the subscript $j$ in this section, and any Nash equilibrium can be represented by a distribution profile $(F_{i})_{i \in [n]}$ on the single battlefield.
Additionally, without loss of generality, we assume $v_{i} = 1$ for every player $i$.

We first show that for any Nash equilibrium, $\forall i \in [n]$, $F_i$ has no mass point in $(0, +\infty)$ (see Lemma \ref{distribution} and Lemma \ref{closed set of supprot}).
Then we establish that, in any Nash equilibrium, $0$ is the common infimum of the supports of all players' strategies, and the common supremum of the supports of all players' strategies is the same (see Lemma \ref{epsilon is supprot} and Lemma \ref{continuous interval}).

The following lemma states that, for any $x > 0$, in any Nash equilibrium, there can be at most one player $i$ whose strategy $F_i(x)$ assigns a non-zero measure at $x$.

\begin{lemma}\label{distribution}
    For any Nash equilibrium $(F_i)_{i \in [n]}$, we have $\# \{i: F_i(x) \neq F_i(x^-)\} \leq 1$, $\forall x > 0$.
\end{lemma}

The next lemma examines the characteristics of the support of players' strategies in a Nash equilibrium.

\begin{lemma}\label{closed set of supprot}
    For any Nash equilibrium $(F_i)_{i \in [n]}$, we have
    \begin{enumerate}
        \item $\bigcup_{i} \overlinex{Supp_{i}}$ is an interval starting from 0.
        \item $\bigcup_{i' \neq i} \overlinex{Supp_{i'}} = \bigcup_{\hat{i}} \overlinex{Supp_{\hat{i}}}$, $\forall i$.
        \item $\forall x \in \bigcup_i \overlinex{Supp_{i}}$, $\# \{i|x \in \overlinex{Supp_i}\} \geq 2$.
    \end{enumerate}
\end{lemma}

Our next lemma states that, for any player, the support of their strategy in a Nash equilibrium either includes the point $\{0\}$ or contains points arbitrarily close to 0.

\begin{lemma}\label{epsilon is supprot}
    Let $\epsilon > 0$, which can be chosen to be arbitrarily small. 
    For any Nash equilibrium $(F_i)_{i \in [n]}$, $\forall i$, we have $[0, \varepsilon] \cap Supp_i \neq \emptyset$.
\end{lemma}

The following lemma is key to the proof of Theorem \ref{NE characterize}.
It asserts that the support of each equilibrium strategy, excluding $\{0\}$, is a single continuous interval, and such continuous intervals of all players share the same right endpoint.
\begin{lemma}\label{continuous interval}
    For any Nash equilibrium $(F_i)_{i \in [n]}$, $\forall i$, we have that there exists $L > 0$ such that $\sup \overlinex{Supp_i} = L$, and $\overlinex{Supp_i} \cap (0, L) = (c_i, L)$, for some $c_i \ge 0$.
\end{lemma}

By Lemma \ref{continuous interval}, let $L$ denote the common supremum of all players' supports, that is, $L = \sup Supp_i$ for every player $i$.
Using Lemmas \ref{distribution} to \ref{continuous interval}, we establish a complete characterization of the structure of Nash equilibrium.
Theorem \ref{NE characterize} states that when there is a single player with the largest budget, the support of this player's equilibrium strategy is a single continuous interval $[0, L]$. 
For the player with the second largest budget, the support of its equilibrium strategy is also a single continuous interval $[0, L]$, but it includes a non-zero measure at $0$. 
For all other players, the supports of their equilibrium strategies all include a non-zero measure at $0$, and are increasing in their budgets in subset relationship, as illustrated in the plot on the left-hand side of Figure \ref{fig-theorem-NE-characterize}.
When multiple players have the largest budget, the supports of their equilibrium strategies are all $[0, L]$. 
For the remaining players, the supports of their equilibrium strategies include a non-zero measure at $0$, and the supports are increasing in their budgets in subset relationship, as depicted in the plot on the right-hand side of the Figure \ref{fig-theorem-NE-characterize}.

\begin{figure}[ht]
    \centering
    \resizebox{0.7 \textwidth}{!}{
    \begin{tikzpicture}
        \node[scale=2.5] at (-2.0,3.5) {Player $1$};
        \node[scale=2.5] at (-2.0,2.5) {Player $2$};
        \node[scale=2.5] at (-2.0,1.5) {Player $3$};
        \node[scale=2.5] at (-2.0,0.5) {Player $4$};
        \draw[fill=black] (-2.0,0) circle (2pt);
        \draw[fill=black] (-2.0,-0.5) circle (2pt);   
        \draw[fill=black] (-2.0,-1) circle (2pt);
        \node[scale=2.5] at (-2.0,-1.5) {Player $n$};

        \node[scale=2.5] at (0,4.2) {$0$};
        \node[scale=2.5] at (7,4.2) {$L$};

        \draw[line width = 1mm] (0,3.5) -- (7,3.5); 
        \draw[fill=blue] (0,2.5) circle (4pt);
        \draw[line width = 1mm] (0,2.5) -- (7,2.5); 
        \draw[fill=blue] (0,1.5) circle (4pt);
        \draw[line width = 1mm] (1,1.5) -- (7,1.5); 
        \draw[fill=blue] (0,0.5) circle (4pt);
        \draw[line width = 1mm] (2,0.5) -- (7,0.5); 
        \draw[fill=blue] (0,-1.5) circle (4pt);
        \draw[line width = 1mm] (4,-1.5) -- (7,-1.5); 

        \node[scale=2.5] at (9.8,3.5) {Player $1$};
        \node[scale=2.5] at (9.8,2.5) {Player $2$};
        \node[scale=2.5] at (9.8,1.5) {Player $3$};
        \node[scale=2.5] at (9.8,0.5) {Player $4$};
        \draw[fill=black] (9.8,0) circle (2pt);
        \draw[fill=black] (9.8,-0.5) circle (2pt);   
        \draw[fill=black] (9.8,-1) circle (2pt);
        \node[scale=2.5] at (9.8,-1.5) {Player $n$};

        \node[scale=2.5] at (11.75,4.2) {$0$};
        \node[scale=2.5] at (18.75,4.2) {$L$};

        \node[scale=2.5] at (1.75, -3) {$B_1 > B_2 > B_3 > B_4 > \cdots > B_n$};

        \node[scale=2.5] at (14.3, -3) {$B_1 = B_2 > B_3 > B_4 > \cdots > B_n$};

        \draw[line width = 1mm] (11.75,3.5) -- (18.75,3.5); 
        \draw[line width = 1mm] (11.75,2.5) -- (18.75,2.5); 
        \draw[fill=blue] (11.75,1.5) circle (4pt);
        \draw[line width = 1mm] (12.75,1.5) -- (18.75,1.5); 
        \draw[fill=blue] (11.75,0.5) circle (4pt);
        \draw[line width = 1mm] (13.75,0.5) -- (18.75,0.5); 
        \draw[fill=blue] (11.75,-1.5) circle (4pt);
        \draw[line width = 1mm] (15.75,-1.5) -- (18.75,-1.5); 
    \end{tikzpicture}
    }
    \caption{The support of Nash equilibrium strategies, blue dot indicates mass point of distribution, black line represents support of distribution. The only difference between the left and right figures is that, in the left figure, player 2 has a mass point at $0$, whereas in the right figure, player 2 does not have a mass point at $0$.}
    \label{fig-theorem-NE-characterize}
\end{figure}

\begin{theorem}\label{NE characterize}
    Relabel the players so that $B_1 \ge B_2 \ge \cdots \ge B_n$, and define $i' \mathrel{\mathop:}= \max \{i: B_i = B_2\} $.
    For any Nash equilibrium $(F_{i})_{i \in [n]}$ and supports $(Supp_i)_{i \in [n]}$ of equilibrium strategies, there exists $L > 0$ such that:
    \begin{enumerate}
        \item If $B_1 > B_2$, then we have
        \begin{enumerate}
            \item $\overlinex{Supp_1} = [0, L]$, $F_1(0) = 0$,
            \item $\overlinex{Supp_i} = [0, L]$ and $F_i(0) > 0$ for all $i \in \{2, 3, \cdots, i'\}$,
            \item $\overlinex{Supp_i} = \{0\} \cup [h_i, L]$ and $F_i(0) > 0$, with $h_{i'+1} > 0$ and $h_i \geq h_{i - 1}$, for all $i \geq i' + 1$.
        \end{enumerate}
        \item If $B_1 = B_2$, then we have
        \begin{enumerate}
            \item $\overlinex{Supp_i} = [0, L]$ and $F_i(0) = 0$ for all $i \in \{1, 2, \cdots, i'\}$.
            \item $\overlinex{Supp_i} = \{0\} \cup [h_i, L]$ and $F_i(0) > 0$, with $h_{i'+1} > 0$ and $h_i \geq h_{i - 1}$, for all $i \geq i' + 1$.
        \end{enumerate}
    \end{enumerate}
\end{theorem}

Theorem \ref{NE characterize} provides the structure of the support of Nash equilibrium strategies for multiple players on a single battlefield.
These structures are crucial for determining Nash equilibrium strategies.
Based on these structures, we can deduce the relationships between the strategies of the players.
We find that, whether $B_1 > B_2$ or $B_1 = B_2$, the strategy of player 2 is closely related to the strategies of the subsequent players.
The following lemma presents two important properties: (1) the strategy of any player $i \in \{3, 4, \cdots, n\}$ is identical to that of player 2 on player $i$'s support set, excluding $\{0\}$, and (2) player 1 and player 2 adopt the same strategy if $B_1 = B_2$.

\begin{lemma}\label{uniqueness-Fi=F2}
    Given the budget vector $\Vec{B} = (B_1, B_2, \cdots, B_n)$ satisfies $B_1 \geq B_2 \geq B_3 \geq \cdots \geq B_n$, then we have that for any player $i \in \{3, 4, \cdots, n\}$, $F_i(x) = F_2(x)$ for all $x \in Supp_i \backslash \{0\}$.
    Additionally, if $B_1 = B_2$, then $F_1(x) = F_2(x)$ for all $x \in Supp_2$.
\end{lemma}

Let $p_i$ denote the probability of the player $i$ bidding 0 and $h_i = \inf (\overlinex{Supp_i \backslash \{0\}})$.
The following lemma describes the relationship between $p_i$ and $h_i$ for player $i \in \{3, 4, \cdots, n\}$.

\begin{lemma}\label{uniqueness-probability-bidding-0}
    For player $i \in \{3, 4, \cdots, n\}$, $F_i(h_i) = F_2(h_i) = p_i$.
\end{lemma}

According to Theorem \ref{NE characterize}, Lemma \ref{uniqueness-Fi=F2} and Lemma \ref{uniqueness-probability-bidding-0}, we can observe that in Nash equilibrium, the strategy of player $i \in [n]$ can be expressed by $F_1$, $F_2$, $h_i$, $L$.
We can set up a system of equations to solve for the Nash equilibrium.
The system can be divided into two cases:
\begin{itemize}
    \item Case (1): $B_1 > B_2 \geq B_3 \geq \cdots \geq B_n$.
    \item Case (2): $B_1 = B_2 \geq B_3 \geq \cdots \geq B_n$.
\end{itemize}

Owing to space constraints, we only present the system of equations corresponding to the Case (2), while the system for Case (1) is included in the Appendix \ref{Equation system}.

Define $Q_i = \prod_{r > i} p_r$.
For Case (2), we have $B_1 = B_2 = \cdots = B_{i'} > B_{i'+1} \geq \cdots \geq B_n$, which implies that $0 < h_{i'+1} \leq \cdots \leq h_n$, $0 = h_1 = \cdots = h_{i'}$, and $p_i > 0$, $\forall i \geq i'+1$.
Therefore, the system of equations is as follows:
\begin{align*}
    & \forall x \in [h_n, L], \forall i \in [n], \,\,
    \begin{cases}
        F_i(x) = (\frac{x}{L})^\frac{1}{n-1}, \\
        \int_{h_n}^L xf_n(x)dx = B_n. 
    \end{cases} \\
    & \forall x \in [h_r, h_{r+1}], \forall r \in \{i'+1, \cdots, n-1\}, \forall i \leq r, \,\,
    \begin{cases}
        F_i(x) = \left[\frac{x}{L Q_{r}}\right]^\frac{1}{r-1}, \\
        \int_{h_r}^{h_{r+1}} xf_r(x)dx = B_r - B_{r+1}.
    \end{cases} \\
    & \forall x \in [0, h_{i'+1}], \forall r \in \{1, 2, \cdots, i'\}, \,\,
    \begin{cases}
        [F_r(x)]^{i'-1} Q_{i'} = \frac{x}{L}, \\
        \int_{0}^{h_{i' + 1}} x f_{i'}(x)dx = B_{i'} - B_{i'+1}
    \end{cases}
\end{align*}

We derive the following corollary based on the system of equations, which applies to the case where there are two players with $B_1 \geq B_2$ competing on a single battlefield.
It is important to note that our corollary aligns with the Nash equilibrium results presented in the paper \citep{r14}.
\begin{corollary}\label{cor:two-player-single-battlefield}
    If there are only two players with $B_1 \geq B_2$ and a single battlefield, then their Nash equilibrium strategies are
    \begin{align*}
        F_1(x) = \frac{x}{2B_1}, \quad x \in [0, 2B_1]; \quad
        F_2(x) = \frac{B_2}{2B_1^2}x + 1 - \frac{B_2}{B_1}, \quad x \in [0, 2B_1].
    \end{align*}
\end{corollary}

\subsection{The Uniqueness of the Nash Equilibrium}
Given that there are at least two players with the maximum budget, i.e., the budget vector $\Vec{B} = (B_1, B_2, \cdots, B_n)$ satisfies $B_1 = B_2 \geq B_3 \geq \cdots \geq B_n$, we prove the uniqueness of the Nash equilibrium.
With the help of the system of equations,
we derive Theorem \ref{theorem-uniqueness-Nash-equilibrium}, which establishes the uniqueness of the Nash equilibrium.

\begin{theorem}\label{theorem-uniqueness-Nash-equilibrium}
    If the budget vector $\Vec{B} = (B_1, B_2, \cdots, B_n)$ satisfies $B_1 = B_2 \geq B_3 \geq \cdots \geq B_n$, then the Nash equilibrium for the GL game with a single battlefield is unique.
\end{theorem}

Assume the budget vector $\Vec{B}$ satisfies $B_1 = B_2 = \cdots = B_{i'} > B_{i'+1} \ge \cdots \ge B_n$.
By Theorem \ref{NE characterize}, we have $0 = h_1 = h_2 = \cdots = h_{i'} < h_{i'+1} \le \cdots \le h_n$, and by Lemma \ref{uniqueness-probability-bidding-0}, $0 = p_1 = p_2 = \cdots = p_{i'} < p_{i'+1} \le \cdots \le p_n$.
Lemmas \ref{uniqueness-Fi=F2} and \ref{uniqueness-probability-bidding-0} together imply that, in Nash equilibrium, every player's strategy can be expressed in terms of player 2's strategy $F_2$.
Hence it suffices to show that player 2's strategy $F_2$ is unique.
Here, we define $M_i = \prod_{i' \geq i} p_{i'}$ and according to the system of equations above, we can observe (1) in the interval $[0, h_{i'+1}]$, player 2's strategy $F_2$ depends only on $L$ and $M_{i'+1}$; (2) for $r \in \{i'+1, \cdots, n-1\}$, in $[h_r, h_{r+1}]$ it depends on $L$ and $M_{r+1}$; (3) in $[h_n, L]$ it depends solely on $L$.
Thus we need only verify the uniqueness of $L$, $M_{i'+1}$ and $M_{r+1}$, $r \in \{i'+1, \cdots, n-1\}$.
Starting from player $n$ and moving backwards, we can compute each $M_i$ and $L$ uniquely from successive budget differences $B_i-B_{i+1}$.
This backward induction shows that $L$ and all relevant $M$-values are uniquely determined; consequently $F_2$ is unique, and therefore the Nash equilibrium itself is unique.

\section{Analysis of Nash Equilibrium with Multiple Battlefields}
In this section, we analyze the properties of the Nash equilibrium in the General Lotto game with multiple battlefields.
When there are multiple battlefields, a player may choose to abandon a battlefield by consistently bidding 0 on it.
We consider two extreme cases:
\begin{itemize}
    \item Suppose that all valuations $ v_{ij} $ are independently and randomly drawn from some continuous distributions, then with probability one, the condition $ \frac{v_{i,j}}{v_{i,j'}} \neq \frac{v_{i',j}}{v_{i', j'}} $ holds for all $i \ne i'$ and $j \ne j'$.
    We find that the average number of battlefields in which each player participates (i.e., bids larger than 0 with positive probability) becomes arbitrarily close to one as $n$ becomes sufficiently large (see Theorem \ref{theorem-MultiBattle-Average-Player}).
    \item In the symmetric setting, where for any player $ i $, $ B_i = B $, and for any two players $ i \neq i' $ and any battlefield $ j $, $ v_{ij} = v_{i'j} = v_j $, we provide a solution for the Nash equilibrium (see Theorem \ref{theorem-NE-multibattlefields}).
\end{itemize}

Finally, we examine intermediate case between these two extremes by providing an example to illustrate the non-uniqueness of the Nash equilibrium.

\begin{theorem}\label{theorem-MultiBattle-Intersection-Player}
    For any Nash equilibrium and any two battlefields $j \ne j'$, let $S$ and $S'$ represent the sets of players who do not always bid 0 on $j$ and $j'$, respectively.
    If $\frac{v_{i,j}}{v_{i,j'}} \neq \frac{v_{i',j}}{v_{i', j'}}$ holds for all $i \ne i'$, then the cardinality of the set $S \cap S'$ is not greater than 3.
\end{theorem}
\begin{proof}
    Let $F_{\cdot, \cdot} = (F_{i,j})_{i \in [n], j \in [m]}$ denote an arbitrary equilibrium.
    By Theorem \ref{NE characterize}, for any $\delta > 0$, every $i \in S$ should bid in $[L_j - \delta, L_j]$ with positive probability on battlefield $j$.
    By Lemma \ref{lem-uniform-converge}, $u_{i,j}(\cdot, F_{-i, j})$ is continuous on $(0, T]$ (the definition of $T$ is given in Lemma \ref{lem-bounded-eq}), then by Lemma \ref{lem-cont-linear} we have $u_{i,j}(L_j, F_{-i, j}) = a_i L_j + b_{i,j}$.

    Consider an arbitrary $i \in S \cap S'$.
    \begin{enumerate}
        \item If $b_{i, j} = b_{i, j'} = 0$, then $u_{i,j}(L_j, F_{-i, j}) = a_i L_j = v_{i,j}$ and $u_{i,j'}(L_{j'}, F_{-i, j'}) = a_i L_{j'} = v_{i,j'}$, therefore $\frac{v_{i,j}}{v_{i,j'}} = \frac{a_i L_j}{a_i L_{j'}} = \frac{L_j}{L_{j'}}$.
        The right-hand side does not depend on $i$.
        Since $\frac{v_{i,j}}{v_{i,j'}} \neq \frac{v_{i', j}}{v_{i', j'}}$, for any given $(j, j')$ there can be at most one such $i$.
        \item If $b_{i, j} = 0$ and $b_{i, j'} > 0$, by Lemma \ref{lem-one-pos-b} we have at most one such $i$ given $j'$.
        \item If $b_{i, j}>0$, again by Lemma \ref{lem-one-pos-b}, we have at most one such $i$ given $j$.
    \end{enumerate}
    We conclude that at most three players in the set $S \cap S'$.
\end{proof}

Based on Theorem \ref{theorem-MultiBattle-Intersection-Player}, we can derive the following theorem.

\begin{theorem}\label{theorem-MultiBattle-Average-Player}
    Let $D_i$ denote the set of battlefields where player $i$ does not always bid 0, and $d_i$ denote the cardinality of the set $D_i$.
    If $\frac{v_{ij}}{v_{ij'}} \neq \frac{v_{i'j}}{v_{i'j'}}$ holds for all pairs of players $i \neq i'$ and all pairs of battlefields $j \neq j'$, we have $\frac{1}{n} \sum_{i = 1}^n d_i < 1 + m \sqrt{\frac{3}{n}}$.
\end{theorem}
\begin{proof}
    Let $S_j$ denote the set of players who do not always bid 0 on battlefield $j$, and $S_{j'}$ denote the set of players who do not always bid 0 on $j'$.
    By Theorem \ref{theorem-MultiBattle-Intersection-Player}, we have $\sum_{(j, j')} |S_j \cap S_{j'}| \leq 3 \cdot \binom{m}{2}$.
    Note that $\sum_{(j, j')} |S_j \cap S_{j'}| = \sum_{i = 1}^n \binom{d_i}{2}$.
    We establish the following inequality
    \begin{align*}
        \sum_{i = 1}^n \binom{d_i}{2} = \sum_{(j, j')} |S_j \cap S_{j'}| \leq 3 \cdot \binom{m}{2} = \frac{3m(m - 1)}{2}.
    \end{align*}
    Using the Cauchy inequality, we have
    \begin{align*}
        \sum_{i = 1}^n \binom{d_i}{2} \geq \frac{1}{2n} (\sum_{i = 1}^n d_i)^2 - \frac{1}{2} \sum_{i = 1}^n d_i.
    \end{align*}
    Therefore, we obtain
    \begin{align*}
        \frac{1}{2n} (\sum_{i = 1}^n d_i)^2 - \frac{1}{2} \sum_{i = 1}^n d_i \leq \frac{3m(m - 1)}{2}.
    \end{align*}
    Let $\zeta = \frac{1}{n} \sum_{i = 1}^n d_i$, we have $n \zeta^2 - n \zeta - 3m(m-1) \leq 0.$
    Solving this inequality, we obtain $\zeta \leq \frac{1}{2} + \sqrt{\frac{1}{4} + \frac{3m(m-1)}{n}}$.
    Furthermore, 
    \begin{align*}
        \sqrt{\frac{1}{4} + \frac{3m(m-1)}{n}} < \frac{1}{2} + m \sqrt{\frac{3}{n}}.
    \end{align*}
    Finally, we have $\frac{1}{n} \sum_{i = 1}^n d_i < 1 + m \sqrt{\frac{3}{n}}$.
\end{proof}

Theorem \ref{theorem-MultiBattle-Average-Player} states that for a fixed number of battlefields, the average number of battlefields in which each player participates can be arbitrarily close to one as $n$ becomes sufficiently large.
Specifically, for $n = m^2$, the average value is less than 2.73.

Although the computation of the Nash equilibria is highly complex in the multi-player and multi-battlefield setting, it can be computed in certain symmetric settings.
The following theorem provides a symmetric solution of the Nash equilibria in a symmetric setting.
\begin{theorem}\label{theorem-NE-multibattlefields}
    Suppose we have $B>0$ such that $B_i = B$ for every $i$, and for every $j$ we have $v_{j}>0$ such that $v_{ij} = v_j$ for every $i$. Then the strategies given by $F_{ij}(x) = (\frac{vx}{v_j n B})^\frac{1}{n-1}$ for every $i$ and $j$, where $v = \sum_{j \in [m]} v_j$, is a Nash equilibrium.
\end{theorem}
\begin{proof}
    Consider the symmetric Nash equilibrium.
    
    Let $e_{ij}$ denote the expected bid of player $i$ on battlefield $j$, and $L_j$ denote the upper bound of the support on the battlefield $j$. 
    Note that all players have the same expected bid on battlefield $j$, and all players play the same strategy, i.e., $e_{ij} = e_j$ and $F_{ij}(x) = F_{j}(x)$ for $\forall i$.
    Therefore, we have
    \begin{align*}
        u_{ij} = v_j (F_{j}(x))^{n-1} = \frac{v_j}{L_j}x.
    \end{align*}
    So we can derive
    \begin{align*}
        \begin{cases}
            F_j(x) = (\frac{x}{L_j})^\frac{1}{n-1}, \\
            f_j(x) = \frac{1}{n-1}\frac{1}{L_j} (\frac{x}{L_j})^\frac{2-n}{n-1}.
        \end{cases}
    \end{align*}

    The expected value of the distribution $F_j(x)$ is equal to the budget. Therefore, we have
    \begin{align*}
        e_j = \int_0^{L_j} x f_j(x) dx = \int_0^{L_j} x \frac{1}{n-1}\frac{1}{L_j} (\frac{x}{L_j})^\frac{2-n}{n-1} dx = \frac{L_j}{n}.
    \end{align*}
    Thus we get $L_j = n e_j$. 
    By Lemma \ref{lem-cont-linear}, we have $\frac{v_1}{L_1} = \frac{v_j}{L_j}$ and $L_j = \frac{v_j}{v_1}L_1$ for $\forall j$.
    We can obtain $n e_j = \frac{v_j}{v_1}L_1$, thereby $e_j = \frac{v_j L_1}{v_1 n}$.
    
    Due to the budget constraints, we have
    \begin{align*}
        \sum_{j \in [m]} e_j = \sum_{j \in [m]} \frac{v_j L_1}{v_1 n} = \frac{L_1 v}{v_1 n} = B,
    \end{align*}
    we get $L_1 = \frac{v_1 n B}{v}$, and $L_j = \frac{v_j n B}{v}$, $\forall j \in [m]$.
    Finally, we have $F_j(x) = \left(\frac{x}{\frac{v_j n B}{v}}\right)^\frac{1}{n-1} = \left(\frac{v x}{v_j n B}\right)^\frac{1}{n-1} = F_{ij}(x)$.
\end{proof}

Following Theorem \ref{theorem-NE-multibattlefields}, we derive the following corollary, which applies to the case where there are two players with $B_1 = B_2$ competing across multiple battlefields where the values of these battlefields are symmetric between the players.
Our corollary \ref{corollary-symmetric} also implies the Nash equilibrium in the symmetric setting established by \citep{r14}.
\begin{corollary}\label{corollary-symmetric}
    If there are only two players with $B_1 = B_2 = B$ and multiple battlefields with $v_{1j} = v_{2j}$ for all $j$, then
    \begin{align*}
        F_{1j}(x) = F_{2j}(x) = \frac{v}{2 B v_j}x, \quad x \in [0, \frac{2 B v_j}{v}],
    \end{align*}
    is a Nash equilibrium.
\end{corollary}

Although we prove in the previous section that Nash equilibrium is unique when at least two players have the maximum budget, this uniqueness does not hold when there are multiple battlefields.
Here is an example.

\paragraph{\textbf{Example.}}
Consider a game with three players and two battlefields.
The players have budgets of $B_1 = 10$, $B_2 = 6$, $B_3 = 6$, respectively, and each player assigns a value of 1 to each battlefield.

For this example, we examine two budget vectors:
\begin{enumerate}[(1)]
    \item $\vec{B}^{(1)} = ((5, 4 ,2), (5, 2, 4))$, in which the expectation of bids on the first battlefield is $(5, 4 ,2)$ corresponding to the resources invested by player 1, 2, and 3 in that battlefield respectively, and the expected bids on the second battlefield is $(5, 2, 4)$.
    \item $\vec{B}^{(2)} = ((5, 3 ,3), (5, 3 ,3))$, in which the expected bids on the first battlefield is $(5, 3 ,3)$ and the expected bids on the second battlefield is $(5, 3 ,3)$.
\end{enumerate}

Let $p_{ij}$ denote the probability that player $i$ bids 0 on battlefield $j$, and let $L_j$ denote the upper endpoint of the players' support on battlefield $j$.

In the first budget vector, it is easy to see that $p_{11} = p_{12}$, $p_{21} = p_{32}$, $p_{31} = p_{22}$ and $L_1 = L_2 = 11.9231$.

In the second budget vector, it follows that $p_{11} = p_{12}$, $p_{21} = p_{22}$, $p_{31} = p_{32}$ and $L_1 = L_2 = 12.5547$.
We observe the symmetry in strategy between the two budget vectors $\vec{B}^{(1)}$ and $\vec{B}^{(2)}$.

Therefore, these strategies indeed form a Nash equilibrium.

\section{Conclusion}
We extend the General Lotto game from a two-player game to a multiplayer game in a general setting, proving the existence of Nash equilibrium in the multiplayer version and providing a detailed characterization of the Nash equilibrium on a single battlefield. 
Additionally, we establish the uniqueness of Nash equilibrium in certain single-battlefield scenarios.
Finally, we generalize the game from a single battlefield to multiple battlefields.
Under the setting of multiple battlefields, we find that the average number of battlefields in which each player participates (i.e., bids larger than 0 with positive probability) becomes arbitrarily close to one as $n$ becomes sufficiently large.
We also show the non-uniqueness of Nash equilibrium by providing an example.

\section*{Acknowledgments}
Zihe Wang was supported by the National Natural Science Foundation of China (Grant No.~62172422), the Fund for Building World-Class Universities (Disciplines) of Renmin University of China, and the Key Laboratory of Interdisciplinary Research of Computation and Economics (Shanghai University of Finance and Economics), Ministry of Education. 
Weiran Shen was supported by the National Natural Science Foundation of China (Grant No.~72192805). Jie Zhang was partially supported by a Leverhulme Trust Research Project Grant (2021--2024) and the EPSRC grant (EP/W014912/1). 
This work was also supported by the Public Computing Cloud of Renmin University of China.

\bibliographystyle{named}
\bibliography{reference}

\newpage
\appendix
\renewcommand{\thesection}{\Alph{section}}
\section{Proofs for Section 2}\label{appendix-sec-2}
\renewcommand{\thesubsection}{\thesection.\arabic{subsection}}
\subsection{Proof of Lemma \ref{lem-cont-linear}}
\begin{proof}
    For every $j \in [m]$, we denote
    \begin{align*}
        \alpha^-_{i,j} \mathrel{\mathop:}=  \inf \{a | \Pr_{X_{i,j} \sim F_{i,j}} \big[ u_{i,j}(x, F_{-i, j}) \le u_{i,j}(X_{i,j}, F_{-i, j}) + a \cdot (x-X_{i,j}), \forall x > X_{i,j} \big] = 1\},
    \end{align*}
    \begin{align*}
        \alpha^+_{i,j} \mathrel{\mathop:}= \sup \{a | \Pr_{X_{i,j} \sim F_{i,j}} \big[ u_{i,j}(x, F_{-i, j}) \le u_{i,j}(X_{i,j}, F_{-i, j}) +  a \cdot (x-X_{i,j}), \forall x < X_{i,j} \big] = 1\}. 
    \end{align*}
    These two notations represent the largest and smallest possible values of the coefficient $a$ if we only focus on battle $j$.
    The lemma is true if
    \[
        \max_{j \in [m]} \alpha^-_{i,j} \le \min_{j \in [m]} \alpha^+_{i,j} < \infty, \forall i \in [n]
    \]
    holds for any Nash equilibrium, because in this case for every $i$ we can pick any $\max_{j \in [m]} \alpha^-_{i,j} \le a_i \le \max_{j \in [m]} \alpha^+_{i,j}$ and compute a corresponding $b_{i,j}$ for every $j \in [m]$. 
    Note that $a_i \le 0$ only if $i$ wins every battlefield w.p. 1, which cannot be the case for any Nash equilibrium, therefore we have $a_i > 0$.
    Since winning probability by bidding 0 is nonnegative, we have $b_{i,j} \ge 0$.
    
    Suppose for some $i \in [n]$, we have
    \[
        \max_{j \in [m]} \alpha^-_{i,j} > \min_{j \in [m]} \alpha^+_{i,j},
    \]
    then there exists $j, j' \in [m]$ and $\tilde a$ s.t. $\alpha^-_{i,j} > \tilde a > \alpha^+_{i,j'}$.
    This means there exists $M, M' \subseteq [0, +\infty)$ and $\tilde x \in [0, +\infty) \backslash M, \tilde x' \in [0, +\infty) \backslash M'$ with $\tilde x \ge \sup M$, $\tilde x' \le \inf M'$ s.t. 
    $F_{i, j}$ assigns positive probability to $M$, $F_{i, j'}$ assigns positive probability to $M'$, and we have
    \begin{align*}
        & \frac{u_{i,j}(\tilde x, F_{-i, j}) - u_{i,j}(x, F_{-i, j})}{\tilde x-x} > \tilde a, \forall x \in M,\\
        & \frac{u_{i,j}(\tilde x', F_{-i, j}) - u_{i,j}(x, F_{-i, j})}{\tilde x'-x} < \tilde a, \forall x \in M'.
    \end{align*}
    Then, suppose $i$ can bid $\tilde x$ on battlefield $j$ with a small probability $\delta > 0$ when it bids in $M$ according to $F_{i,j}$, and bid $\tilde x'$ on $j'$ with a small probability $\delta' > 0$ when it bids in $M'$ according to $F_{i,j'}$. 
    For appropriate $\delta$ and $\delta'$, $i$'s expected total bid does not increase, while its expected utility increases. 
    This contradicts with Nash equilibrium condition.

    Then suppose for some $i \in [n]$, we have 
    \[
        \max_{j \in [m]} \alpha^-_{i,j} \le \min_{j \in [m]} \alpha^+_{i,j}, \forall i \in [n],
    \]
    however $\alpha^+_{i,j} = \infty$ for every $j \in [m]$. 
    Since $u_{i,j}$ is lower bounded by $0$, this holds only if $i$ bids 0 with probability 1 on every battlefield. 
    In this case, $i$'s strategy cannot be a best response.
\end{proof}

\subsection{Proof of Lemma \ref{lem-one-pos-b}}
\begin{proof}
    For given $j$, we consider the infimum value of the player $i$'s bid support on battlefield $j$, denoted by $\underline{Supp_{i,j}}$. W.l.o.g., we have $\underline{Supp_{i,j}}=\inf\{x: F_{i,j}(x)>0\}$. 
    Obviously, we have $u_{i,j}(x,F_{-i,j})=0$ for $x\in [0,\underline{Supp_{i',j}}), \forall i' \ne i$.

    Suppose there are two player $i\neq i'$ such that $\underline{Supp_{i,j}}<\underline{Supp_{i',j}}$, we claim $\underline{Supp_{i,j}}=0$. 
    Actually, we have $u_{i,j}(x,F_{-i,j})=0, \forall x<\underline{Supp_{i',j}}$. 
    For any element $y$ in the set $Supp_{i,j}\bigcap [0, \underline{Supp_{i',j}})$, we have $u_{i,j}(y,F_{-i, j})=0$. 
    By Lemma \ref{lem-cont-linear}, $y$ must be zero. 
    Therefore, $\underline{Supp_{i,j}}=0$.

    We let $K$ denote the set of players whose infimum of support on $j$ is nonzero. 
    If $|K| \ge 2$, by the above argument, all players in $K$ have the same infimum of support. 
    We denote this infimum of support by $\underline{Supp_K}$. 
    Suppose for every player $i\in K$, it holds that player $i$ puts a mass at $\underline{Supp_K}$ in $F_{i,j}$. 
    If $\underline{Supp_K} = T$, then every player only bids $0$ and $T$, and for at least $n-1$ players, the probability of winning by bidding 0 is equal to zero, therefore the claim holds. 
    Otherwise, each player in $K$ would like to deviate by reporting slightly larger than $\underline{Supp_K}$ instead of $\underline{Supp_K}$ to avoid tie-breaking so that its expected utility increases, a contradiction. 
    If there exists one player $i$ who does not put a mass at $\underline{Supp_K}$, we have $u_{i',j}(\underline{Supp_K}^+, F_{-i', j})=0$ for all $i' \neq i$.
    However Lemma \ref{lem-cont-linear} tells that $u_{i',j}(x) = a_{i'}x + b_{i',j}$ for all $x \in Supp_{i',j}$, which contradicts with $\underline{Supp_K} > 0$ and $b_{i',j} \ge 0$.
    
    Therefore we have $|K| \le 1$. If $|K| = 1$, we have $u_{i,j}(x, F_{-i,j}) = 0$ for all $i \in [n] \backslash K$ and all $x \in [0, \underline{Supp_K})$.
    Moreover, we have $[0, \underline{Supp_K}) \cap \underline{Supp_{i,j}} \ne \emptyset$ for all $i \notin K$, therefore $b_{i,j} = 0$ for all $i \notin K$.
    
    If $|K| = 0$, by definition of $K$, there is $|\{i|b_{i,j}>0\}|> 1$ only if every player bids 0 on battlefield $j$.
    In this case, any of them can benefit from bidding slightly larger than 0 to avoid any tie-breaking at 0, which contradicts with the equilibrium condition.
\end{proof}

\section{Proofs for Section 3}\label{appendix-sec-3}
\subsection{Proof of Lemma \ref{lem-gk-nash}}
\begin{proof}
    Consider the set-valued function $\beta_k: \mathcal{F}_k \rightarrow 2^{\mathcal{F}_k}$ given by
    \begin{align*}
        \beta_k(F_{\cdot,\cdot,k}) \mathrel{\mathop:}= \{F'_{\cdot,\cdot,k} | F_{i,\cdot,k}' \in \argmax_{D \in \mathcal{F}_{i, k}} U_i(D, F_{-i,\cdot,k}), \forall i \in [n]\}, \forall F_{\cdot,\cdot,k} \in \mathcal{F}_k.
    \end{align*}
    By its definition, the fixed point of $\beta_k$, which is $F_{\cdot,\cdot,k}$ satisfying $F_{\cdot,\cdot,k} \in \beta_k(F_{\cdot,\cdot,k})$, is a Nash equilibrium of $\mathcal{G}_k$.
    $\beta_k$ satisfies the following properties:
    \begin{enumerate}
        \item $\beta_k$ never takes $\emptyset$: $U_i$ is continuous, and $\mathcal{F}_k$ is a compact subset of $\Delta(A_k)^{mn}$, therefore the maximum in the definition of $\beta_k$ is always attainable for every player $i \in [n]$.
        \item The graph of $\beta_k$ is closed: For any sequence of strategy profiles converging in total variation distance, the corresponding sequence of $(U_i)_{i \in [n]}$ converges in infinite norm. 
        If a limit element in the corresponding output sequence of $\beta$ is not optimal, it has to be sub-optimal for some postfix of the $U_i$ sequence.
    \end{enumerate}
    
    By Kakutani's fixed point theorem, $\beta_k$ has a fixed point which is a Nash equilibrium of $\mathcal{G}_k$, and denote such a fixed point by $\tilde F_{\cdot,\cdot,k} = (\tilde F_{i,\cdot,k})_{i \in [n]}$. 
    As stated above, $\tilde F_{\cdot,\cdot,k}$ is a Nash equilibrium of $\mathcal{G}_k$.
\end{proof}

\subsection{Proof of Lemma \ref{lem-bounded-eq}}
In order to prove Lemma \ref{lem-bounded-eq}, we first state Lemma \ref{lem-linear}, which will be used in the proof of Lemma \ref{lem-bounded-eq}.
Lemma \ref{lem-linear} can be proved in almost the same way as Lemma \ref{lem-cont-linear}.
\begin{lemma}\label{lem-linear}
    $\tilde F_{\cdot,\cdot,k}$ is a Nash equilibrium of $\mathcal{G}_k$ only if there exists $a_i > 0$ for every $i \in [n]$ and $b_{ij} \ge 0$ for every $i \in [n], j \in [m]$ s.t. 
    \[
        \Pr_{x \sim F_{i, j, k}}\left[u_{i,j}(x, \tilde F_{-i, j, k}) = a_{i} x + b_{ij}\right] = 1,
    \]
    and
    \[
        u_{i,j}(x, \tilde F_{-i, j, k}) \le a_{i} x + b_{ij}, \forall x \in A_k
    \]
    holds for every $i \in [n], j \in [m]$.
\end{lemma}

Now we are ready to start proving Lemma \ref{lem-bounded-eq}.
\begin{proof}
    Denote $w \mathrel{\mathop:}= \max_{i \in [n]} B_i$, and then we have $T = 2^{2n+3} w$. 
    Suppose for some equilibrium $(F_{i,\cdot,k})_{i \in [n]}$ there exists some player $i'$ who bids in $[\frac{T}{2} + \frac{T}{k}, T]$ on some battlefield $j$ with positive probability. 
    Then there exists some player $i'' \ne i'$ bidding in $[\frac{T}{2}, T]$ on battlefield $j$ with positive probability, otherwise player $i'$'s utility is not maximized by bidding in $[\frac{T}{2} + \frac{T}{k}, T]$.
    
    We claim there exists a player $\hat i \in \{i', i''\}$ with $u_{\hat i, j}(2w, F_{-\hat i, j, k}) < \frac{v_{\hat i, j}}{2^{2n}}$. 
    If this is not true, combining $u_{\hat i, j}(2w, F_{-\hat i, j, k}) \ge \frac{v_{\hat i, j}}{2^{2n}}$ with (i) $u_{\hat i, j}(x, F_{-\hat i, j, k}) \le a_{\hat i} x + b_{\hat i, j}$ for all $x$, (ii) $u_{\hat i, j}(x, F_{-\hat i, j, k}) = a_{\hat i} x + b_{\hat i, j}$ for all $x$ that $\hat i$ bids with positive probability, and (iii) that $\hat i$ bids in $[T/2, T]$ with positive probability, gives
    \begin{align*}
        & a_{\hat i} \cdot T/2 + b_{\hat i, j} \le {v_{\hat i, j}}, \\
        & a_{\hat i} \cdot 2w + b_{\hat i, j} \ge \frac{v_{\hat i, j}}{2^{2n}}. \\
        \Rightarrow & b_{\hat i, j} \ge \frac{v_{\hat i, j}}{2^{2n}} - \frac{2w}{T/2-2w} \cdot ({v_{\hat i, j}}-\frac{v_{\hat i, j}}{2^{2n}}) = \frac{v_{\hat i, j}}{2^{2n+1}}.
    \end{align*}
    Therefore $b_{\hat i, j}$ is lower bounded by a constant that does not depend on $k$.
    There are two cases: (i) At least one of $i'$ and $i''$ bids 0 on battlefield $j$ with positive probability. 
    In this case, for this player, randomly bidding 0 and in $[T/2, T]$ is worse than bidding $T/k$, for any sufficiently large $k$. 
    (ii) Neither $i'$ nor $i''$ bids 0 on battlefield $j$ with positive probability. One can easily show that there exists some $c_k > 0$ s.t. for every player, the minimum nonzero bid that it chooses for battlefield $j$ with positive probability is equal to $\frac{c_kT}{k}$.
    By $u_{i, j}(0, F_{-i, j, k}) \ge 0$ and $u_{i,j}(T, F_{-i, j, k}) \ge \frac{v_{i, j}}n$, we have $u_{i,j}(\frac{c_kT}{k}) \ge \frac{c_k v_{i, j}}{nk}$ for every $i \in [n]$.

    For every $\hat i \in \{i', i''\}$, by $b_{\hat i, j} \ge \frac{v_{\hat i, j}}{2^{2n+1}}$ and $u_{\hat i,j}(T, F_{-\hat i, j, k}) \ge \frac{v_{\hat i, j}}n$ we have
    \begin{align*}
        u_{\hat i,j}(\frac{c_kT}{k}) \ge \frac{T-\frac{c_kT}{k}}{T} \cdot \frac{v_{\hat i, j}}{2^{2n+1}} + \frac{\frac{c_kT}{k}}{T} \cdot \frac{v_{\hat i, j}}n = (\frac{k-c_k}{2^{2n+1}k} + \frac{c_k}{nk}){v_{\hat i, j}}.
    \end{align*}
    Note that since $i'$ and $i''$ never bid smaller than $\frac{c_kT}{k}$, we have $u_{\hat i, j}(x, F_{-\hat i, j, k}) = 0$ for all $x < \frac{c_kT}{k}$. 
    $i'$ and $i''$ can only win at $\frac{c_kT}{k}$ by winning a tie-breaking, therefore by deviating to $\frac{(c_k+1)T}{k}$, the winning probability at least doubles. 
    This gives
    \begin{align*}
         a_{\hat i} \ge \frac{\frac12 \cdot (\frac{k-c_k}{2^{2n+1}k} + \frac{c_k}{nk}){v_{\hat i, j}}}{\frac{T}{k}} = \frac{\frac{k-c_k}{2^{2n+1}} + \frac{c_k}{n}}{2T} \cdot {v_{\hat i, j}} \ge \frac{k}{2^{2n+2}T} \cdot {v_{\hat i, j}}.
    \end{align*}
    That is, $a_{\hat i}$ becomes unbounded as $k$ increases, which gives unbounded value of $u_{\hat i, j'}$ for some $j' \in [m]$, a contradiction.

    Therefore, our initial claim is true, that is, there exists some player $\hat i \in \{i', i''\}$ with $u_{\hat i, j}(2w, F_{-\hat i, j, k}) < \frac{v_{\hat i, j}}{2^{2n}}$. 
    However, this means there exists some $i^* \ne \hat i$ who bids smaller than $2w$ with probability less than $\frac12$. 
    Then $i^*$'s expected bid is larger than $w$, which violates its budget constraint since $w \ge B_{i^*}$. We conclude that no player bids in $[T/2+w, T]$ with positive probability.
\end{proof}

\subsection{Proof of Lemma \ref{lem-uniform-converge}}
In order to prove Lemma \ref{lem-uniform-converge}, we first state Lemma \ref{lem-zero-cont} and Lemma \ref{lem-continuous}, which will be used in the proof of Lemma \ref{lem-uniform-converge}.

\begin{lemma}\label{lem-zero-cont}
    Given $n \ge 3$, for any limit $\tilde F_{\cdot, \cdot}$ of the subsequence $(\tilde F_{\cdot, \cdot, k})_{k \ge 1}$, any $i \in [n]$ and any $j \in [m]$, we have $\tilde F_{i,j}(0) = \tilde F_{i,j}(0^+)$.
\end{lemma}
\begin{proof}
    Suppose there exists $i,j$ s.t. $\tilde F_{i,j}(0) \ne \tilde F_{i,j}(0^+)$. 
    By Helly's selection theorem $\tilde F$ is non-decreasing, therefore $\tilde F_{i,j}(0^+) > \tilde F_{i,j}(0)$. Denote $\delta \mathrel{\mathop:}= \frac{1}{3}(\tilde F_{i,j}(0^+) - \tilde F_{i,j}(0))$. 
    Take a sequence $(\epsilon_d)_{d \ge 1}$ in $(0, T)$ that decreases to zero, for every $d \ge 1$ there exists $K_d > 0$ s.t. for every $k > K_d$, we have 
    \[
        \tilde F_{i,j,k}(\epsilon_d) \ge \tilde F_{i,j}(\epsilon_d) - \delta \ge \tilde F_{i,j}(0^+) - \delta, 
    \]
    and
    \[
        \tilde F_{i,j,k}(0) \le \tilde F_{i,j}(0) + \delta, 
    \]
    therefore 
    \[
        \tilde F_{i,j,k}(\epsilon_d) - \tilde F_{i,j,k}(0) \ge \delta.
    \]
    That is, with probability no less than $\delta$, player $i$'s bid on battlefield $j$ is in $(0, \epsilon_d]$. Denote by $w$ and $W$ the minimum and maximum positive bid within $(0, \epsilon_d]$ that player $i$ chooses on battlefield $j$, so that $\tilde F_{i,j,k}(W) = \tilde F_{i,j,k}(\epsilon_d)$. 
    We claim that every $i' \ne i$ should bid less than $w$ with positive probability: Otherwise, instead of bidding $w$, $i$ can either bid $0$ without decreasing its utility (if some $i'$ always bids larger than $w$), or bid $w + \frac{T}{k}$ so that its winning probability $w + \frac{T}{k}$ more than doubles (if every $i' \ne i$ bids $w$ with positive probability). 
    In both case such deviation of $i$ can be coupled with a larger or smaller bid so that $i$'s expected utility increases.

    When bidding less than $w$, player $i' \ne i$'s winning probability is upper bounded by $\tilde F_{i,j,k}(0) \cdot \Pi_{i'' \ne i, i'} \tilde F_{i'',j,k}(W)$.
    When bidding $W+\frac{T}{k}$, its winning probability is lower bounded by $\tilde F_{i,j,k}(W) \cdot \Pi_{i'' \ne i, i'} \tilde F_{i'',j,k}(W)$. 
    Therefore, the constant $a_{i'}$ given in Lemma \ref{lem-linear} is lower bounded:
    \begin{align*}
        a_{i'} & \ge \frac{(\tilde F_{i,j,k}(W) - \tilde F_{i,j,k}(0)) \cdot \Pi_{i'' \ne i, i'} \tilde F_{i'',j,k}(W)}{W+\frac{T}{k}} \cdot v_{i', j} \\
        & \ge \frac{\delta \cdot \Pi_{i'' \ne i, i'} \tilde F_{i'',j,k}(W)}{W+\frac{T}{k}} \cdot v_{i', j} \\
        & \ge \frac{\delta \cdot \Pi_{i'' \ne i, i'} \tilde F_{i'',j,k}(W)}{2W} \cdot v_{i', j}.
    \end{align*}
    Note that with positive probability, $i'$ bids no less than $\frac{B_{i'}}{m}$ on some battlefield $j'$.
    Therefore, we have
    \begin{align*}
        a_{i'} \cdot \frac{B_{i'}}{m} \le v_{i', j} \Rightarrow \Pi_{i'' \ne i, i'} \tilde F_{i'',j,k}(W) \le \frac{2mW}{\delta \cdot B_{i'}}.
    \end{align*}
    The above inequality holds for every $i' \ne i$, which further gives
    \begin{align*}
        & \Pi_{i' \ne i} \Pi_{i'' \ne i, i'} \tilde F_{i'',j,k}(W) \le (\frac{2mW}{\delta})^{n-1}\cdot \frac{1}{\Pi_{i' \ne i} B_{i'}}\\
        \Rightarrow & \Pi_{i'' \ne i} \tilde F_{i'',j,k}(W) \le (\frac{2mW}{\delta})^{\frac{n-1}{n-2}} \cdot (\Pi_{i' \ne i} B_{i'})^{-\frac{1}{n-2}}.
    \end{align*}
    By Lemma \ref{lem-linear}, we have
    \begin{align*}
        &a_iW + b_{ij} \le (\frac{2mW}{\delta})^{\frac{n-1}{n-2}} \cdot (\Pi_{i' \ne i} B_{i'})^{-\frac{1}{n-2}}  \cdot v_{i', j}\\
        \Rightarrow & a_i \le (\frac{2m}{\delta})^{\frac{n-1}{n-2}} \cdot (\Pi_{i' \ne i} B_{i'})^{-\frac{1}{n-2}} \cdot v_{i', j} \cdot W^{\frac{1}{n-2}}.
    \end{align*}
    Therefore, $a_i$ goes to zero as $d$ goes to infinity and $W \le \epsilon_d$ goes to zero, which means for sufficiently large $k$, the marginal contribution to $i$'s winning probability as it increases its bid goes to zero on every battlefield. 
    This is true only if on every battlefield, $i$ is almost always among the highest bidders, which easily leads to a contradiction since other players can beat $i$ on battlefield $j$.
\end{proof}

\begin{lemma} \label{lem-continuous}
    For any limit $\tilde F_{\cdot, \cdot}$ of the subsequence $(\tilde F_{\cdot, \cdot, k})_{k \ge 1}$, any $i \in [n]$ and any $j \in [m]$, $\tilde F_{i,j}$ is continuous on $(0, T)$.
\end{lemma}
\begin{proof}
    Note that every $\tilde F_{i,j}$ is non-decreasing, therefore its left and right limit both exist at every $b \in (0, T)$. 
    The lemma is true if the following two claims both hold:
    \begin{enumerate}
        \item There does not exist $j \in [m]$ and  $b \in (0, T)$ s.t. for some $i, i' \in [n]$ with $i \ne i'$, $\tilde F_{i,j}(b^+) > \tilde F_{i,j}(b^-)$ and $\tilde F_{i',j}(b^+) > \tilde F_{i',j}(b^-)$;
        \item There does not exist $j \in [m]$ and $b \in (0, T)$ s.t. for some $i \in [n]$, $\tilde F_{i,j}(b^+) > \tilde F_{i,j}(b^-)$ and $\tilde F_{i',j}(b^+) = \tilde F_{i',j}(b^-)$ for all $i' \ne i$;
    \end{enumerate}
    We give proof for the two claims separately.
    \begin{enumerate}
            \item Suppose there exists such $(\tilde S, i, i', j, b)$. 
            Denote
            \begin{align*}
                \delta_1 \mathrel{\mathop:}= \frac{\tilde F_{i,j}(b^+) - \tilde F_{i,j}(b^-)}3,  \quad \delta_2 \mathrel{\mathop:}= \frac{\tilde F_{i',j}(b^+) - \tilde F_{i',j}(b^-)}3.
            \end{align*}
            Take a decreasing sequence $(\epsilon_d)_{d \ge 1}$ taking values in $(0, \min($ $b, T-b))$ that goes to zero. 
            For every $d \ge 1$, by monotonicity of $F$ and pointwise convergence of $(\tilde F_{\cdot, \cdot, k})$, there exists $K_{d,0} > 0$ s.t. for any $k \ge K_{d,0}$,
            \[
                \tilde F_{i,j,k}(b - \epsilon_d) \le \tilde F_{i,j}(b - \epsilon_d) + \delta_1 \le \tilde F_{i,j}(b^+) + \delta_1.
            \]
            Similarly, for every $d \ge 1$, there exists $K_{d,1} > 0$ s.t. for any $k \ge K_{d,1}$,
            \[
                \tilde F_{i,j,k}(b + \epsilon_d) \ge \tilde F_{i,j}(b + \epsilon_d) - \delta_1 \ge \tilde F_{i,j}(b^+) - \delta_1.
            \]
            Consider the sequence
            $$(K_d)_{d \ge 1} \mathrel{\mathop:}= (\max\{K_{d,0},K_{d,1}\})_{d \ge D},$$ 
            which satisfies, for any $k \ge K_d$,
            \begin{align*}
                \tilde F_{i,j,k}(b + \epsilon_d) - \tilde F_{i,j,k}(b - \epsilon_d) > \tilde F_{i,j}(b^+) - \delta_1 - \tilde F_{i,j}(b^-) - \delta_1 > \delta_1.
            \end{align*}
            Similarly, for player $i'$, we obtain a sequence $(K_d')_{d \ge 1}$ such that for any $k \ge K_d',$
            \[
                \tilde F_{i',j,k}(b + \epsilon_d) - \tilde F_{i',j,k}(b - \epsilon_d) > \delta_2.
            \]
            Consider the sequence
            $$(\tilde K_d)_{d \ge 1} \mathrel{\mathop:}= (\max\{K_d, K'_d\})_{d \ge 1}.$$
            For any $\epsilon > 0$, there exists $d > 0$ s.t. $\epsilon_{d} < \epsilon$, and for any $k > \tilde K_{d}$,
            \begin{align*}
                    & \tilde F_{i, j, k}(b+\epsilon) - \tilde F_{i, j, k}(b - \epsilon) \ge \tilde F_{i,j,k}(b+\epsilon_d) - \tilde F_{i,j,k}(b - \epsilon_{d}) > \delta_1, \\
                    & \tilde F_{i', j,k}(b+\epsilon) - \tilde F_{i', j,k}(b - \epsilon) \ge \tilde F_{i',j,k}(b+\epsilon_d) - \tilde F_{i',j,k}(b - \epsilon_{d}) > \delta_2.
            \end{align*}
        \begin{enumerate}
            \item For any $\epsilon < min(b,T-b)$ and any corresponding $\epsilon_d < \epsilon$, $k>\tilde K_d$, suppose there exists some player who bids larger than $b+\epsilon$ on battlefield $j$ w.p. $1$ by strategy profile $\tilde F_{\cdot,\cdot,k}$. 
            Then, $i$ and $i'$ both bid within range $(b-\epsilon, b+\epsilon]$ with positive probability, but never wins. 
            This is strictly worse than randomly bidding $0$ and $T$, a contradiction with $\tilde F_{\cdot, \cdot, k}$ being an equilibrium of $\mathcal{G}_k$.
            \item Therefore for any $\epsilon > 0$ and any correspondingly large $k$, we have $\tilde F_{i'', j,k}(b+\epsilon) > 0, \forall i'' \in [n]$. 
            Denote by $\mathcal{E}$ the event that both player $i$ and $i'$ bid within range $(b-\epsilon, b+\epsilon]$, and all other players bid at most $b+\epsilon$, we have $\Pr[\mathcal{E}] \ge \delta_1 \cdot \delta_2 \cdot \Pi_{i'' \ne i, i'}\tilde F_{i'', j,k}(b+\epsilon)$. 
            Moreover, denote by $\mathcal{E}_i$ and $\mathcal{E}_{i'}$ the subset of $\mathcal{E}$ that $i$ and $i'$ separately wins. 
            By pigeonhole principle, w.l.o.g. we can suppose 
            \[
                \Pr[\mathcal{E}_i] \le \frac12 \Pr[\mathcal{E}].
            \]
            Moreover, we have
            \begin{align*}
                \Pr[i \text{ wins by bidding in }(b-\epsilon,b+\epsilon] \le \Pi_{i'' \ne i, i'}\tilde F_{i'', j,k}(b+\epsilon).
            \end{align*}

            For sufficiently large $k$, we have $A_k \cap (b+\epsilon, b+2\epsilon) \ne \emptyset$. 
            In such a game $\mathcal{G}_k$, whenever $i$ used to bid within $(b - \epsilon, b+\epsilon]$ according to $\tilde F_{i, j,k}$, it can bid within range $(b+\epsilon, b+2\epsilon)$ w.p. at least $\frac{b-\epsilon}{b+2\epsilon}$, and bid $0$ otherwise. 
            Player $i$'s expected bid does not increase, while the change in probability of winning is lower bounded by
            \begin{align*}
                & \frac{b-\epsilon}{b+2\epsilon} \cdot \left( \Pr \big[ i \text{ wins by bidding in }(b-\epsilon,b+\epsilon] \big]
                + \Pr[\mathcal{E} \backslash \mathcal{E}_i] \right) \\ & - \Pr \big[ i \text{ wins by bidding in }(b-\epsilon,b+\epsilon] \big] \\
                \ge & \frac{b-\epsilon}{b+2\epsilon} \cdot \frac12 \delta_1 \cdot \delta_2 \cdot \Pi_{i'' \ne i, i'}\tilde F_{i'',j, k}(b+\epsilon) - \frac{3\epsilon}{b+2\epsilon} \Pi_{i'' \ne i, i'}\tilde F_{i'',j, k}(b+\epsilon).
            \end{align*}
            By taking a sufficiently small $\epsilon$, we can make the above value positive, which contradicts with $\tilde F_{i,\cdot, k}$ being a best response.
        \end{enumerate}
        \item Suppose there exists such $(\tilde F, i, j, b)$. 
        Denote $\delta = \frac{\tilde F_{i}(b^+) - \tilde F_{i}(b^-)}3$.
        For any $\xi > 0, e > 0$ there exists $\epsilon \in (0, e)$ and $K> 0$ s.t. for any $k > K$,
        \begin{align*}
            \tilde F_{i, j, k}(b+\epsilon) - \tilde F_{i, j, k}(b - \epsilon) > \delta, \quad \tilde F_{i', j, k}(b+\epsilon) - \tilde F_{i', j, k}(b - \epsilon) < \xi, ~ \forall i' \ne i.
        \end{align*}
        Obviously we have $\tilde F_{i', j, k}(b+\epsilon) > 0$ for all $i' \ne i$, otherwise $\tilde F_{i, \cdot, k}$ is not $i$'s best response.
        Consider two choices of some player $i' \ne i$ on battlefield $j$:
        \begin{enumerate}
            \item Bid $b - \epsilon - c$ for some $c \ge 0$. 
            \item Bid $b + \epsilon'$ for some $\epsilon' \in (\epsilon, 2\epsilon] \cap A_k$ w.p. $\frac{b - \epsilon - c}{b+2\epsilon}$, and bid 0 otherwise.
        \end{enumerate}
    
        The difference in expected utility between these choices is 
        \begin{align*}
            \Delta_P = \; & u_{i',j}(b - \epsilon - c, \tilde F_{-i', j, k}) - \frac{3\epsilon + c}{b+2\epsilon} \cdot u_{i'}(0, \tilde F_{-i', j, k}) \\
            &- \frac{b - \epsilon - c}{b+2\epsilon} \cdot u_{i',j}(b + \epsilon', \tilde F_{-i', j, k}) \\
            \le \; & \tilde F_{i, j, k}(b-\epsilon) \cdot \prod_{i'' \ne i, i'} \tilde F_{i'', j, k}(b+\epsilon) - \frac{b - \epsilon - c}{b+2\epsilon} \\
            &\cdot \big[ \tilde F_{i, j, k}(b-\epsilon) + \delta \big] \cdot \prod_{i'' \ne i,i'} \tilde F_{i'', j, k}(b+\epsilon) \\
            = \; & \Big[\frac{3\epsilon + c}{b + 2\epsilon} \tilde F_{i, j, k}(b-\epsilon) - \frac{b - \epsilon - c}{b+2\epsilon} \cdot \delta \Big] \cdot \prod_{i'' \ne i,i'} \tilde F_{i'', j, k}(b+\epsilon) \\
            < \; & \Big[ \frac{3\epsilon + c}{b + 2\epsilon}(1 - \delta) - \frac{b - \epsilon - c}{b+2\epsilon} \cdot \delta \Big] \cdot \prod_{i'' \ne i,i'} \tilde F_{i'', j, k}(b+\epsilon)
        \end{align*}
    
        We have $\Delta_P < 0$ if $c \le \delta (b + 2\epsilon) - 3\epsilon$. 
        Therefore, for any $c \le \delta b - 3\epsilon < \delta (b + 2\epsilon) - 3\epsilon$ the first option is strictly worse for player $i' \ne i$, and the second option has weakly smaller bid expectation. 
        Since $\tilde F_{\cdot, \cdot, k}$ is an equilibrium, no player $i' \ne i$ bids within range $[(1-\delta) b + 2\epsilon, b - \epsilon]$ on battlefield $j$.
        This gives 
        \begin{equation}\label{eq-c-interval}
            \tilde F_{i', j, k}((1-\delta)b + 2\epsilon) = \tilde F_{i', j, k}(b-\epsilon), \forall i' \ne i.
        \end{equation}
        
        Now consider two choices of player $i$:
        \begin{enumerate}
            \item Bid within $((1-\delta)b + 2\epsilon, (1-\delta)b + 3\epsilon] \cap A_k$ w.p. $\frac{T-b+\epsilon}{T-(1-\delta)b-3\epsilon}$, and bid $T$ w.p. $\frac{\delta b - 4\epsilon}{T-(1-\delta)b-3\epsilon}$.
            \item Bid within range $(b-\epsilon, b+\epsilon] \cap A_k$.
        \end{enumerate}

        Note that player $i$ takes the second choice with positive probability according to $\tilde F_{i,\cdot,k}$. 
        The first choice has smaller bid expectation, therefore its  winning probability has to be weakly smaller. 
        This is true only if:
        \begin{align*}
            & \frac{T-b+\epsilon}{T-(1-\delta)b-3\epsilon} \cdot u_{i, j} \left( (1-\delta) b + 2\epsilon, \tilde F_{-i, j, k} \right) \\
            & + \frac{\delta b - 4\epsilon}{T-(1-\delta)b-3\epsilon} \cdot u_{i,j}(T, \tilde F_{-i, j, k}) \\
            \le \; & u_{i,j}(b, \tilde F_{-i, j, k}) / v_{ij} \\
            \Rightarrow & \frac{T-b+\epsilon}{T-(1-\delta)b-3\epsilon} \cdot \Pi_{i' \ne i} \tilde F_{i', j, k} \left( (1-\delta) b + 2\epsilon \right) \\
            & + \frac{\delta b - 4\epsilon}{T-(1-\delta)b-3\epsilon} \cdot u_{i, j}(T, \tilde F_{-i, j, k}) \\
            \le \; & \Pi_{i' \ne i} \tilde F_{i',j,k}(b+\epsilon).
        \end{align*}
        We have already obtained
        \[
            \tilde F_{i',j,k}(b+\epsilon) < \tilde F_{i',j,k}(b-\epsilon) + \xi, \forall i' \ne i.
        \]
        And recall equation (\ref{eq-c-interval}), we have:
        \begin{align*}
            & \Pi_{i' \ne i} \tilde F_{i',j,k}(b+\epsilon) \\
            <~ & \Pi_{i' \ne i} (\tilde F_{i',j,k}((1-\delta) b + 2\epsilon) + \xi) \\
            <~ & \Pi_{i' \ne i} \tilde F_{i',j,k}((1-\delta) b + 2\epsilon) + n\xi \\
            \Rightarrow~ & \frac{\delta b - 4\epsilon}{T-(1-\delta)b-3\epsilon} \cdot [u_{i, j}(T, \tilde F_{-i, j, k}) - \Pi_{i' \ne i} \tilde F_{i', j, k}((1-\delta) b + 2\epsilon)]
            <~ n\xi 
        \end{align*}
        Considering any tie of $i$ and some $i' \ne i$ at $T$, and any case where some $i' \ne i$ bids within $((1-\delta) b + 2\epsilon ,T)$, we have
        \begin{align*}
            n \xi >~ & \frac{\delta b - 4\epsilon}{T-(1-\delta)b-3\epsilon} \cdot \Big(\Pr[\max_{i' \ne i}X_{i',j} \in ((1-\delta) b + 2\epsilon, T)] \\
            & + \sum_{2 \le l \le n} \frac1l \Pr[l-1 \text{ players other than } i \text{ bid }T] \Big) \\
            \ge~ & \frac{\delta b - 4\epsilon}{T-(1-\delta)b-3\epsilon} \cdot \frac1n \cdot \Pr \Big[\max_{i' \ne i} X_{i',j} \in ((1-\delta) b + 2\epsilon, T] \Big]
        \end{align*}
        for any $\tilde F_k$ with sufficiently large $k$. 
        Since we can take arbitrarily small $\xi$, the probability that any $i' \ne i$ bidding larger than $(1-\delta)b + 2\epsilon$ can be arbitrarily small given sufficiently large $k$. 
        In this case, however, $i$'s winning probability should be arbitrarily close to $1$ on every battlefield, otherwise it should bid smaller on $j$ to gain higher utility on some other battlefield, which leads to contradiction.
    \end{enumerate}
\end{proof}

Now we are ready to start proving Lemma \ref{lem-uniform-converge}.
\begin{proof}
    \begin{enumerate}
        \item By Lemma \ref{lem-zero-cont} and \ref{lem-continuous}, we know that every $\tilde F_{i,j}$ is continuous at $[0, T)$. 
        Moreover, by Lemma \ref{lem-bounded-eq}, we have $\tilde F_{i,j,k}(\frac{3}{4}T) = 1$ for any sufficiently large $k$, which gives $\tilde F_{i,j}(\frac{3}{4}T) = 1$. 
        Helly's selection theorem tells that $\tilde F_{i,j}$ is non-decreasing, so that $\tilde F_{i,j}(x) = 1$ for all $x \in [\frac{3}{4}T, T]$. 
        Then we have that $\tilde F_{i,j}$ is continuous on the closed interval $[0,T]$.
        By Heine-Cantor theorem, we can conclude that $\tilde F_{i,j}$ is uniformly continuous.
        \item We have obtained that $\tilde F_{i,j}$ is continuous. 
        Then, $(\tilde F_{i,j,k})_{k \ge 1}$ is a sequence of non-decreasing functions defined on closed interval $[0,T]$ that pointwise converges to a continuous limit $\tilde F_{i,j}$, therefore the convergence is uniform (see e.g. \cite{uniform-converge}).
        \item First, note that for every $x \in [0,T]$ and set of c.d.f. $F_{-i,j}$, the value of $u_{i,j}(x, F_{-i, j})$ only depends on $F_{i',j}(x)$ and $F_{i', j}(x^-)$ for every $i' \ne i$. 
        Moreover, these $2n-2$ values are in $[0,1]$, and one can easily show that for every $x$, the value of $u_{i,j}(x, F_{-i, j})$ is Lipschitz continuous in all these values.

        For any $i \in [n]$, $j \in [m]$ and $\epsilon > 0$, by uniform convergence there exists $K > 0$ s.t. for any $k > K$, we have $|\tilde F_{i', j, k}(x) - \tilde F_{i',j}(x)| < \epsilon$ for every $i' \ne i$ and every $x \in [0, T]$. 
        By uniform continuity of $\tilde F_{i',j}$ there exists $\delta > 0$ s.t. $\tilde F_{i', j}(x-\delta) \ge \tilde F_{i', j}(x^-) - \epsilon$ for every $i' \ne i$, every $j \in [m]$ and every $x \in [0, T]$. 
        Therefore, for any $k \ge K$, we have
        \begin{align*}
            \tilde F_{i', j}(x^-) - 2\epsilon & \le \tilde F_{i', j}(x-\delta) - \epsilon  \le \tilde F_{i', j, k}(x-\delta) \\
            & \le \tilde F_{i', j, k}(x^-) \le \tilde F_{i', j, k}(x) \\
            & \le \tilde F_{i', j}(x) + \epsilon \\
            & = \tilde F_{i', j}(x^-) + \epsilon,
        \end{align*}
        and 
        \[
            \tilde F_{i', j}(x) - \epsilon \le \tilde F_{i', j, k}(x) \le \tilde F_{i', j}(x) + \epsilon,
        \]
        for every $i' \ne i$ and every $x \in [0,T]$. 
        By Lipschitz continuity of $u_{i,j}$, one can easily obtain the convergence of $u_{i,j}$ as a function of $x$ in infinite norm.
    \end{enumerate}
\end{proof}

\subsection{Proof of Theorem \ref{thm-mbeq-exist}}
\begin{proof}
     We have obtained the sequence of Nash equilibria $(\tilde F_{\cdot,\cdot,k})_{k \ge 1}$ of the sequence of games $(\mathcal{G}_k)_{k \ge 1}$, which converges to a limit $\tilde F_{\cdot,\cdot}$. 
    We claim that $\tilde F_{\cdot,\cdot}$ is a Nash equilibrium of the original GL game with continuous and unbounded bid space.

    First we show that $\tilde F_{\cdot,\cdot}$ is a Nash equilibrium if the bid space of each player is continuous, but constrained to be $[0, T]$. 
    Suppose $\tilde F_{\cdot,\cdot}$ is not a Nash equilibrium, then there exists $i \in [n]$ who can deviate to another strategy $F_{i,\cdot} = (F_{i, j})_{j \in [m]}$, so that the difference in expected utility is positive:
    \begin{align*}
        \delta =& \sum_{j \in [m]} \mathbb{E}_{X_{i,j} \sim F_{i,j}}[u_{i,j}(X_{i,j}, \tilde F_{-i, j})] - \sum_{j \in [m]} \mathbb{E}_{X_{i,j} \sim \tilde F_{i,j}}[u_{i,j}(X_{i,j}, \tilde F_{-i, j})] \\
        >&~ 0.
    \end{align*}
    By lemma \ref{lem-uniform-converge}, there exists $K > 0$ s.t. for any $k > K$, we have
    $\sup_{x \in [0,T]}|u_{i, j}(x, \tilde F_{-i,j,k}) - u_{i,j}(x, \tilde F_{-i, j})| < \frac{\delta}{3}$ for every $x \in [0, T]$ and every $j \in [m]$. Therefore, we have
    \begin{align*}
        & \sum_{j \in [m]} \mathbb{E}_{X_{i,j} \sim F_{i,j}}[u_{i,j}(X_{i,j}, \tilde F_{-i, j, k})] - \sum_{j \in [m]} \mathbb{E}_{X_{i,j} \sim \tilde F_{i,j}}[u_{i,j}(X_{i,j}, \tilde F_{-i, j, k})] \\
        > & \sum_{j \in [m]} \mathbb{E}_{X_{i,j} \sim F_{i,j}}[u_{i,j}(X_{i,j}, \tilde F_{-i, j}) -  \frac{\delta}{3}] - \sum_{j \in [m]} \mathbb{E}_{X_{i,j} \sim \tilde F_{i,j}}[u_{i,j}(X_{i,j}, \tilde F_{-i, j}) + \frac{\delta}{3}] \\
        = & ~ \frac{\delta}{3}.
    \end{align*}
    That is, when other players bid according to $\tilde F_{-i, \cdot, k}$, bidding according to $F_{i,\cdot}$ gives player $i$ higher utility than bidding according to $\tilde F_{i,\cdot}$. 
    Note that $u_{i,j}(\cdot, \tilde F_{-i, j, k})$ is a piecewise constant function, therefore we can modify every $X_{i,j} \sim F_{i,j}$ into $X_{i,j}' \mathrel{\mathop:}= \lceil \frac{X_{i,j} \cdot k}{T} \rceil \cdot \frac{T}{k}$, so that (i) $x_{i,j}$' only takes value in $A_k$, (ii) $i$'s expected utility does not decrease, while (iii) its total bid expectation increases by at most $\frac{T}{k}$. 
    We can further construct $X_{i,j}'' = X_{i,j}' \cdot I$ for every $j \in [m]$, where $I$ is a random number independent of $(X_{i,j}')_{j \in [m]}$ that equals $0$ with probability $p$ and equals $1$ with probability $(1-p)$.
    For sufficiently large $k$, we can choose a small enough $p$ so that $\sum_{j \in [m]} \mathbb{E}[X_{i,j}''] \le B_i$, and
    \begin{align*}
        & \sum_{j \in [m]} \mathbb{E}[u(X_{i,j}'', \tilde F_{-i, j, k})] \\
        \ge & ~ (1-p) \cdot \sum_{j \in [m]} \mathbb{E}_{X_{i,j} \sim F_{i,j}}[u_{i,j}(X_{i,j}, \tilde F_{-i, j, k})] \\
        \ge & ~ \sum_{j \in [m]} \mathbb{E}_{X_{i,j} \sim F_{i,j}}[u_{i,j}(X_{i,j}, \tilde F_{-i, j, k})] - \frac{\delta}{6}.
    \end{align*}
    Lemma \ref{lem-uniform-converge} tells that $(\tilde F_{i, j, k})_{k \ge 1}$ uniformly converges to $\tilde F_{i, j}$ for every $j \in [m]$. 
    Moreover, we have
    \begin{align*}
        & \mathbb{E}_{X_{i,j} \sim \tilde F_{i,j}}[u_{i,j}(X_{i,j}, \tilde F_{-i, j, k})] \\ 
        = &~ \int_{0}^{T} (1-\tilde F_{i,j}(x)) du_{i,j}(x, \tilde F_{-i, j, k}) \\
        = &~ u_{i,j}(T, \tilde F_{-i, j, k}) - u_{i,j}(0, \tilde F_{-i, j, k}) - \int_{0}^{T} \tilde F_{i,j}(x) du_{i,j}(x, \tilde F_{-i, j, k}).
    \end{align*}
    For every distribution profile $F_{-i, j, k}$, $u_{i,j}(\cdot, F_{-i, j, k})$ is non-decreasing and takes value in $[0, v_{i,j}]$. 
    By taking some sufficiently large $k$, we can let $\tilde F_{i, j, k}$ be arbitrarily close to $\tilde F_{i, j}$ for every $x \in [0, T]$ and every $j \in [m]$ simultaneously, so that
    \begin{align*}
        & \sum_{j \in [m]} \mathbb{E}_{X_{i,j} \sim \tilde F_{i,j,k}}[u_{i,j}(X_{i,j}, \tilde F_{-i, j, k})] \\
        = & \sum_{j \in [m]} \mathbb{E}_{X_{i,j} \sim \tilde F_{i,j}}[u_{i,j}(X_{i,j}, \tilde F_{-i, j, k})] + \sum_{j \in [m]} \int_0^T (\tilde F_{i,j,k}(x) - \tilde F_{i,j}(x)) du_{i,j}(x, \tilde F_{-i, j, k}) \\
        \le & \sum_{j \in [m]} \mathbb{E}_{X_{i,j} \sim \tilde F_{i,j}}[u_{i,j}(X_{i,j}, \tilde F_{-i, j, k})] + \frac{\delta}{6} \\
        < & \sum_{j \in [m]} \mathbb{E}_{X_{i,j} \sim F_{i,j}}[u_{i,j}(X_{i,j}, \tilde F_{-i, j, k})] - \frac{\delta}{6}.
    \end{align*}

    So far, we have identified a deviation characterized by random bids $(X_{i,j}'')_{j \in [m]}$ in $\mathcal{G}_k$, which yields higher expected utility than $\tilde F_{i,\cdot,k}$.
    This contradicts with $\tilde F_{\cdot,\cdot,k}$ being a Nash equilibrium of $\mathcal{G}_k$. 
    We conclude that deviation $F_{i,\cdot}$ in the game with continuous bid space $[0, T]$ should not exist.

    Therefore, $\tilde F_{\cdot,\cdot}$ is a Nash equilibrium of the game with bid space $[0,T]$. 
    However, by Lemma \ref{lem-bounded-eq} and the pointwise convergence of $\tilde F_{\cdot,\cdot,k}$ to $\tilde F_{\cdot,\cdot}$, no player bid larger than $\frac{3}{4}T$ in $\tilde F_{\cdot,\cdot}$. 
    Therefore, even if a player is allowed to bid larger than $T$, any deviation bidding larger than $T$ with positive probability does not induce higher expected utility than another deviation that only bids in $[0,T]$. 
    We conclude that $\tilde F_{\cdot,\cdot}$ is a Nash equilibrium of the original GL game.
\end{proof}

\section{Proofs for Section 4}\label{appendix-sec-4}
\subsection{Proof of Lemma \ref{distribution}}
\begin{proof}
    We prove it by contradiction.
    Define $\Gamma_x = \{i':F_{i'}(x) \neq F_{i'}(x^-)\}$ and suppose $|\Gamma_x| \geq 2$ for some $x > 0$.
    Take an arbitrary $i \in \Gamma_x$, we have $x \in Supp_i$.
    By Lemma \ref{lem-cont-linear}, $a_i(x + \varepsilon) + b_i \geq u_i(x + \varepsilon, F_{-i})$ holds for every $\varepsilon > 0$.
    By taking $\lim_{\varepsilon \rightarrow 0}$, we obtain following inequality
    \begin{align*}
        a_ix + b_i \geq u_i(x^+, F_{-i}) = \prod_{i' \neq i}F_{i'}(x) \geq u_i(x^-, F_{-i}) + c = a_ix + b_i + c,
    \end{align*}
    where $c = \frac{\prod_{i' \in \Gamma_x}(F_{i'}(x) - F_{i'}(x^-))}{|\Gamma_x|} \cdot \prod_{i'' \notin \Gamma_x}F_{i''}(x)$.
    To satisfy the inequality there should be $c =0$. 
    However, it means $\prod_{i'' \notin \Gamma_x}F_{i''}(x) = 0$ by definition of $\Gamma_x$, therefore $u_i(x^+, F_{-i}) = \prod_{i'' \neq i}F_{i''}(x) = 0$.
    Note that we have $u_i(x, F_{-i}) = a_ix+b_i \le u_i(x^+, F_{-i})$, then by $x>0$ we obtain $a_i = b_i = 0$, which means player $i$ never wins. This easily leads to a contradiction.
    We conclude that $|\Gamma_x| < 2$ for every $x > 0$.
\end{proof}

\subsection{Proof of Lemma \ref{closed set of supprot}}
\begin{proof}
    For (1), we prove it by contradiction.
    $\exists (x_1, x_2)$ such that $\forall i, (x_1, x_2) \cap \overlinex{Supp_i} = \emptyset$ and $x_2 \leq sup \bigcup_i Supp_i$.
    We make $x_2$ as large as possible satisfying $x_2 \in \bigcup_i \overlinex{Supp_i}$.
    W.l.o.g., we can assume $x_2 \in \overlinex{Supp_1}$.
    By Lemma \ref{distribution}, at most one player bids $x_2$ with positive probability at $x_2$.
    If such player exists, then the player is denoted by $i'$, otherwise, $i'$ is refered to any player and $x_2 \in \overlinex{Supp_{i'}}$.
    Because of the continuity of the $u_{i'}(x, F_{-{i'}})$ at $x = x_2$, we can get that $\forall \varepsilon > 0$, $[x_2, x_2 + \varepsilon) \cap Supp_{i'} \neq \emptyset$.
    Furthermore, $\exists \hat{x} \in [x_2, x_2 + \varepsilon)$, $u_{i'}(\hat{x}, F_{-{i'}}) = a_{i'} \hat{x} + b_{i'}$.
    When $\hat{x} \rightarrow x_2$, we have $u_{i'}(x_2, F_{-i'}) = a_{i'}x_2 + b_{i'}$.
    On the other hand, $\forall i, (x_1, x_2) \cap \overlinex{Supp_i} = \emptyset$, and $x \in (x_1, x_2], u_{i'}(x, F_{-i'}) = a_{i'}x_2 + b_{i'}$, which contradicts $u_{i'}(x,F_{-i'}) \leq a_{i'}x + b_{i'}$ in Lemma \ref{lem-cont-linear}.

    For (2), we prove it by contradiction.
    Suppose $\exists i, x_1, x_2$ satisfying $(x_1, x_2) \subseteq \overlinex{Supp_i}$ and $\forall i' \neq i,  (x_1, x_2) \cap \overlinex{Supp_{i'}} = \emptyset$, then $u_i(x, F_{-i})$ is constant for $x \in (x_1, x_2)$, which contradicts $u_i(x, F_{-i}) = a_ix + b_i, x \in Supp_i$ in Lemma \ref{lem-cont-linear}.

    For (3), if $x \in \bigcup_i \overlinex{Supp_i}$, then $\exists i'', x \in \overlinex{Supp_{i''}}$.
    According to the (2), $x \in \bigcup_{i' \neq i''} \overlinex{Supp_{i'}}$, then we can have $\exists l \neq i'', x \in \overlinex{Supp_l}$.
    Therefore, $i'', l \in \{i|x \in \overlinex{Supp_i}\}$.
\end{proof}

\subsection{Proof of Lemma \ref{epsilon is supprot}}
\begin{proof}
    Consider positive $x \in \bigcup_i \overlinex{Supp_i}$.
    By Lemma \ref{closed set of supprot}, $\exists i', i''$ satisfying $x \in \overlinex{Supp_{i'}} \cap \overlinex{Supp_{i''}}$.
    One of the following cases holds:
    \begin{itemize}
        \item Case 1: $\forall \delta > 0, \exists x^* \in (x - \delta, x] \cap Supp_{i'}$.
        \item Case 2: $\forall \delta > 0, \exists x^* \in (x, x + \delta) \cap Supp_{i'}$.
    \end{itemize}
    In case 1, we have $u_i(x^*) = a_ix^* + b_i$ and $u_i(x) \geq u_i(x^*) > 0$.
    In case 2, $u_i(x^*) = a_ix^* + b_i$.
    When $\delta \rightarrow 0$, we can get that $u_i(x^+) = a_ix + b_i > 0$ and $u_i(x) \geq \frac{u_i(x^+)}{n} > 0$.
    In the both cases, $u_i(x) > 0$, so $\prod_{i \neq i'}F_i(x) > 0$. Similarly, $\prod_{i \neq i''}F_i(x) > 0$.
    Hence, $\forall i, F_i(x) > 0$.
\end{proof}

\subsection{Proof of Lemma \ref{continuous interval}}
\begin{proof}
    We prove it by contradiction.
    There exists $i, \underline{\alpha}, \overline{\alpha}$ satisfying (1) $\underline{\alpha_i} \in \overline{Supp_i}$, (2) $\underline{\alpha_i} > 0$, (3) $\overline{\alpha_i} \leq L$, and (4) $(\underline{\alpha_i}, \overline{\alpha_i}) \cap \overline{Supp_i} = \emptyset$.
    Note, we don't require $\overline{\alpha_i}$ to be in the $\overlinex{Supp_i}$.
    For $\forall i' \neq i$, one of the following holds:
    \begin{itemize}
        \item Case 1: $(\underline{\alpha_i}, \overline{\alpha_i}) \subseteq \overlinex{Supp_{i'}}$.
        \item Case 2: $(\underline{\alpha_i}, \overline{\alpha_i}) \cap \overlinex{Supp_{i'}} = \emptyset$.
        \item Case 3: $\exists (\underline{\alpha_{i'}}, \overline{\alpha_{i'}}] \subseteq (\underline{\alpha_i}, \overline{\alpha_i}]$, $\underline{\alpha_{i'}} \in \overlinex{Supp_{i'}}$ and $(\underline{\alpha_{i'}}, \overline{\alpha_{i'}}) \cap \overlinex{Supp_{i'}} = \emptyset$.
    \end{itemize}
    
    For Case 3, we consider $(\underline{\alpha_{i'}}, \overline{\alpha_{i'}})$ instead of $(\underline{\alpha_{i}}, \overline{\alpha_{i}})$.
    Finally, we can find an interval $(\underline{\alpha_{i^*}}, \overline{\alpha_{i^*}})$ such that (1) $\underline{\alpha_{i^*}} \in \overlinex{Supp_{i^*}}$, (2) $\underline{\alpha_{i^*}} > 0$, (3) $\overline{\alpha_{i^*}} \leq L$, and (4) $(\underline{\alpha_{i^*}}, \overline{\alpha_{i^*}}) \cap \overlinex{Supp_{i''}} = \emptyset$ or $(\underline{\alpha_{i^*}}, \overline{\alpha_{i^*}}) \subseteq \overlinex{Supp_{i''}}$ for $\forall i''$.

    Let $S = \{i''|(\underline{\alpha_{i^*}}, \overline{\alpha_{i^*}}) \subseteq \overlinex{Supp_{i''}}\}$.
    By Lemma \ref{closed set of supprot}, we have $|S| \geq 2$.
    Next, we assert that player $i^*$ has motivation to bid $\underline{\alpha_{i^*}} + \varepsilon$.

    Let $S' = [n] \backslash S$, $s = |S|$.
    Clearly, for $\forall i \in S'$, $x \in (\underline{\alpha_{i^*}}, \underline{\alpha_{i^*}} + \varepsilon)$, $F_i(x) = F_i(\underline{\alpha_{i^*}})$.
    Let $\lambda = \prod_{i \in S'}F_i(x)$, $\lambda' = \prod_{i \in S' \backslash \{i^*\}}F_i(x)$.
    We only consider the case of the utility of the player $i \in S$ is $u_i(x, F_{-i}) = a_ix$.
    Other cases are similar to it.
    For $\forall i \in S$, we have $u_i(x, F_{-i}) = a_ix = \prod_{i' \in S'}F_{i'}(x) \cdot \prod_{i' \in S \backslash \{i\}}F_{i'}(x) = \lambda \cdot \prod_{i' \in S \backslash \{i\}}F_{i'}(x)$.
    Thereby, $\prod_{i \in S}a_ix = \lambda^s \cdot (\prod_{i \in S}F_i(x))^{s - 1}$, and we can have $\prod_{i \in S}F_i(x) = (\frac{\prod_{i \in S}a_i}{\lambda^s})^{\frac{1}{s - 1}}x^{\frac{s}{s - 1}}$.
    Then if $i^*$ bid $x \in (\underline{\alpha_{i^*}}, \underline{\alpha_{i^*}} + \varepsilon)$, we can calculate the utility, that is $u_{i^*}(x, F_{-i^*}) = (\frac{\prod_{i \in S}a_i}{\lambda^s})^{\frac{1}{s - 1}}x^{\frac{s}{s - 1}} \cdot \lambda'$.
    Let $\gamma = (\frac{\prod_{i \in S}a_i}{\lambda^s})^{\frac{1}{s - 1}}\lambda'$ and $t = \frac{s}{s - 1}$, respectively, which are constants.
    Therefore, $u_{i^*}(x, F_{-i^*}) = \gamma \cdot x^t$.
    
    We take the derivative of $x$ in the interval $(\underline{\alpha_{i^*}}, \underline{\alpha_{i^*}} + \varepsilon)$ and obtain $\gamma t x^{t - 1}$.
    Due to the continuity at point $x = \underline{\alpha_{i^*}}$, we have $a_{i^*} \cdot \underline{\alpha_{i^*}} = \gamma \cdot \underline{\alpha_{i^*}}^t$.
    Then $a_{i^*} = \frac{\gamma \cdot \underline{\alpha_{i^*}}^t}{\underline{\alpha_{i^*}}} = \gamma \cdot \underline{\alpha_{i^*}}^{t - 1}$.
    We can obtain the following inequality in the interval $x \in (\underline{\alpha_{i^*}}, \underline{\alpha_{i^*}} + \varepsilon)$
    \begin{equation}
        \gamma t x^{t - 1} \geq \gamma t \underline{\alpha_{i^*}}^{t - 1} > \gamma \underline{\alpha_{i^*}}^{t - 1} = a_{i^*}.
    \end{equation}
    From the inequality above, we can know that when $x \in (\underline{\alpha_{i^*}}, \underline{\alpha_{i^*}} + \varepsilon)$, the derivative is higher than that when $x = \underline{\alpha_{i^*}}$.
    Therefore, $i^*$ will bid $\underline{\alpha_{i^*}} + \varepsilon$ rather than $\underline{\alpha_{i^*}}$, which shows that it is impossible to have such an interval $(\underline{\alpha_{i^*}}, \overline{\alpha_{i^*}})$.
\end{proof}

\subsection{Proof of Theorem \ref{NE characterize}}
To prove Theorem \ref{NE characterize}, we first present Lemma \ref{u is positive}.
The lemma states that a positive bid is guaranteed to generate a positive utility in any Nash equilibrium.
\begin{lemma}\label{u is positive}
    For any Nash equilibrium $(F_i)_{i \in [n]}$, we have $u_i(x, F_{-i}) > 0$, $\forall i$, $\forall x > 0$.
\end{lemma}
\begin{proof}
    By Lemma \ref{closed set of supprot}, there exists $i' \ne i''$ s.t. $x \in \overlinex{Supp_{i'}} \cap \overlinex{Supp_{i''}}$. 
    By Lemma \ref{lem-cont-linear}, we have $u_{i'}(x', F_{-i'}) = \prod_{i \neq i'}F_{i}(x') = a_{i'}x' + b_{i'} > 0$ for every $x' \in Supp_{i'}$.
    This together with $x > 0$ gives $F_{i}(x) > 0$ for every $i \ne i'$, because every $F_i$ is non-decreasing, and $Supp_{i'}$ is a dense subset of $\overlinex{Supp_{i'}}$.
    Similarly, we have $F_{i}(x) > 0$ for every $i \ne i''$.
    We conclude that $F_i(x) > 0$, and $u_i(x, F_{-i}) > 0$  for all $i$ and $x > 0$.
\end{proof}

Now, we begin to prove Theorem \ref{NE characterize}.
\begin{proof}
    By Lemma \ref{epsilon is supprot} and Lemma \ref{continuous interval}, for any $i$, we have either (a) $\overlinex{Supp_i} = [0, L]$ or (b) $\overlinex{Supp_i} = [h_i, L] \cup \{0\}$ for some $h_i > 0$.
    For case (b), we must have $F_i(0) > 0$; otherwise, for any player $i^* \neq i$ and for any $x \in (0, h_i)$, we would have $u_{i^*}(x, F_{-i^*}) = 0$, which contradicts Lemma \ref{u is positive}.
    
    For Case 1, we assume that $u_i(x, F_{-i}) = a_ix$ holds for every $i$.
    We can easily show that $a_i = \frac{1}{L}$ for every $i$.
    Define index $(z_i)_{i \in [n]}$ satisfying $\overlinex{Supp_{z_1}} \supseteq \overlinex{Supp_{z_2}} \supseteq \overlinex{Supp_{z_3}} \supseteq \cdots$, then $u_{z_i}(x, F_{-z_i}) = u_{z_{i + 1}}(x, F_{-z_{i + 1}})$ holds for all $x \in Supp_{z_{i + 1}} \backslash \{0\}$ and all $i < n$.
    It implies $F_{z_{i + 1}}(x) = F_{z_{i}}(x)$ for all $x \in Supp_{z_{i + 1}} \backslash \{0\}$ and all $i < n$.
    By Lemma \ref{distribution} we can obtain that $\mathbb{E}_{X \sim F_{z_{i + 1}}}[X] \leq \mathbb{E}_{X \sim F_{z_{i}}}[X]$, which further gives $B_{z_{i + 1}} \le B_{z_i}$. 
    If in addition we have $\overlinex{Supp_{z_{i + 1}}} = \overlinex{Supp_{z_i}}$, then $\mathbb{E}_{X \sim F_{z_{i + 1}}}[X] = \mathbb{E}_{X \sim F_{z_{i}}}[X]$, so that $B_{z_{i + 1}} = B_{z_i}$.
    When $B_1$ is the unique maximum budget, there has to be $\overlinex{Supp_1} \supsetneq \overlinex{Supp_i}$ for all $i \neq 1$, hence we have $\bigcup_{i' \neq 1} \overlinex{Supp_{i'}} \neq \bigcup_{i} \overlinex{Supp_i}$ which contradicts with the second item in Lemma \ref{closed set of supprot}.
    So our assumption is wrong.
    That is, there exists $i \in [n]$ s.t. $u_i(x, F_{-i}) = a_ix + b_i$ with $b_i > 0$, and we have $a_iL + b_i = 1$.
    By Lemma \ref{lem-one-pos-b}, we can obtain that for any $i^* \neq i$, $u_{i^*} = \frac{1}{L}x$.
    Note that $u_i(x, F_{-i}) = a_i x + b_i$ and $u_{i^*}(x, F_{-i^*}) = \frac{1}{L} x$, in which $a_i = \frac{1 - b_i}{L}$, so we have $a_i < \frac{1}{L}$ in the interval $[0, L]$.
    Then we have $u_i(x, F_{-i}) > u_{i^*}(x, F_{-i^*})$ for all $x < L$, and we can obtain that $F_{i^*}(x) > F_i(x)$, $\mathbb{E}_{X \sim F_{i}}[X] > \mathbb{E}_{X \sim F_{i^*}}[X]$, $B_i > B_{i^*}$, and $i = 1$.
    Therefore, (a) holds.
    As argued above, for any $i', i'' \neq 1$, $\overlinex{Supp_{i'}} \subsetneq  \overlinex{Supp_{i''}}$ only if $ B_{i'} < B_{i''}$, therefore (c) holds.
    Furthermore, for $\overlinex{Supp_2} = \overlinex{Supp_3} = \cdots = \overlinex{Supp_{i'}}$, due to Lemma \ref{closed set of supprot}, $\overlinex{Supp_2} = [0, L]$, and (b) holds.

    For Case 2, we assume that $\exists i$ such that $u_i(x, F_{-i}) = a_ix + b_i$, $b_i > 0$, which implies that for any $i^* \neq i$, $u_{i^*}(x, F_{-{i^*}}) = \frac{1}{L} x$.
     Since $u_i(x, F_{-i}) > u_{i^*}(x, F_{-{i^*}})$ for all $x < L$, we can obtain that $F_{i^*}(x) > F_i(x)$, $\mathbb{E}_{X \sim F_{i}}[X] > \mathbb{E}_{X \sim F_{i^*}}[X]$, $B_i > B_{i^*}$, and $i = 1$, which contradicts $B_1 = B_2$.
     So our assumption is wrong.
     That is $u_i(x, F_{-i}) = \frac{1}{L} x$ holds for every $i$.
     We similarly define index $(z_i)_{i \in [n]}$ s.t. $B_{z_i} \geq B_{z_{i+1}}$.
     As argued above, we can know that $\overlinex{Supp_{z_i}} = \overlinex{Supp_{z_{i+1}}}$ implies $B_{{z_i}} = B_{{z_{i+1}}}$ and $F_{z_i} = F_{z_{i+1}}$.
     In addition, $\overlinex{Supp_{z_i}} \supsetneq \overlinex{Supp_{z_{i+1}}}$ implies $B_{{z_i}} > B_{{z_{i+1}}}$.
     Hence, we have $\overlinex{Supp_i} = [0, L]$, $\forall i \in \{1, 2, \cdots, i' \}$, and $\overlinex{Supp_i} = \{0\} \cup [h_i, L]$, $F_i(0) > 0$, $i \geq i'+1$.
     Therefore, (b) holds.
     Because of $u_1(x) = \frac{1}{L} x$, then $\prod_{i > 1} F_i(0) = 0$. 
     We can deduce $\prod_{2 \leq i \leq n} F_i(0) = 0$. 
     Thus, $F_i(0) = 0$, $\forall i \in \{1, 2, \cdots, i'\}$.
     Therefore, (a) holds.
\end{proof}

\subsection{Proof of Lemma \ref{uniqueness-Fi=F2}}
\begin{proof}
    Consider two players $1 < i < i'$ so that $B_i \geq B_{i'}$, and $h_i = \inf (\overlinex{Supp_i\backslash\{0\}})$, $h_{i'} = \inf (\overlinex{Supp_{i'}\backslash\{0\}}$), by Theorem \ref{NE characterize}, we have $h_i \leq h_{i'}$.
    For any $x \in (h_{i'}, L]$, we have
    \begin{align*}
        u_i(x) &= F_1(x) \cdot \prod_{i'' \in [n] \backslash \{1, i\}} F_{i''}(x) = \frac{1}{L} x,  \\
        u_{i'}(x) &= F_1(x) \cdot \prod_{i'' \in [n] \backslash \{1, i'\}} F_{i''}(x) = \frac{1}{L} x,
    \end{align*}
    which implies $F_{i'}(x) = F_i(x)$ for all $x \in (h_{i'}, L]$.
    Since this holds for any $1 < i < i' \le n$, we can obtain that for any player $i \in \{2, 3, \cdots, n\}$, we have $F_i(x) = F_{i'}(x)$ for any $i' > i$ and any $x \in [h_{i'}, L]$. 
    In particular, $F_2(x) = F_i(x)$ for every $i \in \{3, 4, \cdots, n\}$ and every $x \in Supp_i \backslash \{0\}$.
    
    If $B_1 = B_2$, by Theorem \ref{NE characterize}, we have $u_1(x) = u_2(x) = \frac{x}{L} = \prod_{i \ne 1} F_i(x) = \prod_{i \ne 2} F_i(x)$ for every $x \in [0, L]$, therefore $F_1(x) = F_2(x)$ for every $x \in [0, L]$.
\end{proof}

\subsection{Proof of Lemma \ref{uniqueness-probability-bidding-0}}
\begin{proof}
    By Lemma \ref{uniqueness-Fi=F2}, for player $i$, $x \in [h_i, L]$, we have $F_i(x) = F_2(x)$. 
    Let $x = h_i$, we have $F_i(h_i) = F_2(h_i)$. 
    If $h_i = 0$, it means $F_i(0) = p_i = 0$, then $F_i(0) = F_2(0) = 0$. 
    If $h_i > 0$, we have $F_i(h_i) = p_i$, then $F_2(h_i) = p_i$.
\end{proof}

\subsection{Equation system for solving Nash equilibrium}\label{Equation system}
Let $Q_i = \prod_{r > i} p_r$.
For Case (1), we have $B_1 > B_2 = \cdots = B_{i'} > B_{i'+1} \geq \cdots \geq B_n$, which implies that $h_1 = h_2 = \cdots = h_{i'} = 0$, $0 < h_{i'+1} \leq h_{i'+2} \leq \cdots \leq h_n$, and $p_i > 0$, $\forall i \geq 2$, $0 < p_2 = \cdots = p_{i'}$.
Therefore, the system of equations is as follows:

\begin{flalign*}
    &\forall x \in [h_n, L], \forall i \in \{2, 3, \cdots, n\}, \,\,
    \begin{cases}
        F_1(x) [F_2(x)]^{n-2} = \frac{x}{L}, \\
        [F_2(x)]^{n-1} = a_1 x + b_1, \\
        \int_{h_n}^L x f_2(x)dx = B_n, \\
        F_2(x) = F_i(x).
    \end{cases}
    &
\end{flalign*}

\begin{flalign*}
    &\forall x \in [h_r, h_{r+1}], \forall r \in \{i'+1, \cdots, n-1\}, \forall i \in \{2, \cdots, r\}, \,\,
    \begin{cases}
        F_1(x) [F_2(x)]^{r-2} Q_{r} = \frac{x}{L}, \\
        [F_2(x)]^{r-1} Q_{r} = a_1 x + b_1, \\
        \int_{h_r}^{h_{r+1}} xf_2(x)dx = B_r - B_{r+1}, \\
        F_2(x) = F_i(x).
    \end{cases}
    &
\end{flalign*}

\begin{flalign*}
    &\forall x \in (0, h_{i'+1}], \forall i \in \{2, 3, \cdots, i'\}, \,\,
    \begin{cases}
        F_1(x) [F_2(x)]^{i'-2} Q_{i'} = \frac{x}{L}, \\
        [F_2(x)]^{i'-1} Q_{i'} = a_1 x + b_1, \\
        Q_1 = b_1, \\
        F_2(x) = F_i(x).
    \end{cases}
    &
\end{flalign*}

\begin{flalign*}
    &\begin{cases}
        a_1 L + b_1 = 1, \\
        \int_0^L x f_1(x)dx = B_1.
    \end{cases}
    &
\end{flalign*}

\subsection{Proof of Corollary \ref{cor:two-player-single-battlefield}}
\begin{proof}
    Because of $B_1 \geq B_2$, let $u_1(x) = a_1 x + b_1 = F_2(x)$, $b_1 \geq 0$, and $u_2(x) = a_2 x = F_1(x)$.
    When $x = L$, we have $a_2 L = F_1(L) = 1$, then we can get $a_2 = \frac{1}{L}$.
    The $L$ can be calculated by the equation $\int_0^L \frac{1}{L}x dx = B_1$ and we can get $L = 2B_1$.
    Moreover, we have $u_1(x) = a_1 x + b_2 = F_2(x)$.
    When $x = L$, we have $a_1 L + b_1 = F_2(L) = 1$, then we can get $b_1 = 1 - a_1 L$.
    The $a_1$ can be calculated by the equation $\int_0^L x d(a_1 x + b_1) = B_2$.
    We substitute $L = 2B_1$ into this equation and obtain $a_1 = \frac{B_2}{2B_1^2}$, thereby we can get $b_1 = 1 - \frac{B_2}{B_1}$.
    Therefore, $F_1(x) = \frac{x}{2B_1}$ and $F_2(x) = \frac{B_2}{2B_1^2}x + 1 - \frac{B_2}{B_1}$, $x \in [0, 2B_1]$.
\end{proof}

\subsection{Proof of Theorem \ref{theorem-uniqueness-Nash-equilibrium}}
\begin{proof}
    By Theorem \ref{NE characterize}, for Player $i \in \{3, 4, \cdots, n\}$ we have $\overlinex{Supp_i} = {0} \cup [h_i, L]$ where $h_i \geq 0$.
    For $i \in \{1, 2\}$ we have $\overlinex{Supp_i} = [0, L]$.
    Moreover, since $p_1 = p_2 = 0$, we have $M_1 = M_2 = 0$.
    
    For player 1, if he bids $x \in [h_i, h_{i+1}]$ with $h_i < h_{i+1}$ for some $i \geq 2$, by Lemma \ref{uniqueness-Fi=F2}, its utility is given by
    \begin{align*}
        u_1(x, F_{-1}) = F_2(x) F_3(x) \cdots F_n(x) = F_2(x)^{i - 1} M_{i + 1} = \frac{1}{L} x.
    \end{align*}
    Therefore, we have $F_2(x) = (\frac{1}{L M_{i + 1}}x)^\frac{1}{i - 1}$ for all $x \in [h_i, h_{i+1}]$.
    For $i = n$ and $x \in [h_n, L]$, we have $u_1(x, F_{-i}) = F_2(x)^{n-1} = \frac{1}{L} x$, therefore $F_2(x) = (\frac{1}{L} x)^\frac{1}{n-1}$.
    This allows us to represent player 2's equilibrium strategy as:
    \begin{equation}\label{NE-uniqueness-B1=B2-F2}
        \begin{aligned}
            F_2(x) = \begin{cases}
                (\frac{1}{L M_{i + 1}}x)^\frac{1}{i - 1}, &x \in [h_i, h_{i+1}], i \in \{2, \cdots, n-1\};  \\ 
                (\frac{1}{L} x)^\frac{1}{n-1}, &x \in [h_n, L].
            \end{cases}
        \end{aligned}
    \end{equation}
    For $3 \leq i < n$, we have
    \begin{align*}
        u_1(h_i, F_{-i}) = F_2(h_i)^{i - 1} M_{i + 1} = p_i^{i - 1} M_{i + 1} = (\frac{M_i}{M_{i+1}})^{i - 1} M_{i + 1} = \frac{1}{L} h_i,
    \end{align*}
    which then gives $h_i = \frac{(M_i)^{i-1}}{(M_{i+1})^{i-2}} L$.
    Then, we can calculate $h_i$ for every $i \in [n]$ using $(M_i)_{i \in [n]}$ and $L$, that is
    \begin{align}\label{eq-hi}
        h_i = \begin{cases}
            0, &i \in \{1,2\};  \\
            \frac{(M_i)^{i-1}}{(M_{i+1})^{i-2}} L, &i \in \{3,4,\cdots,n-1\};  \\
            (M_n)^{n-1} L, &i = n.
        \end{cases}
    \end{align}
    By Lemma \ref{uniqueness-Fi=F2}, for $i \in \{2,3,\cdots,n-1\}$, $B_i - B_{i+1}$ can be given by
    \begin{equation*}
        \begin{aligned}
            & B_i - B_{i+1} \\
            = & \int_{h_i}^{h_{i+1}} xf_2(x)dx  \\
            = &~ xF_2(x) \Big|_{h_i}^{h_{i+1}} - \int_{h_i}^{h_{i+1}} F_2(x)dx  \\
            = &~ h_{i+1} F_2(h_{i+1}) - h_i F_2({h_i}) - \left[(1 - \frac{1}{i}) L M_{i+1}(\frac{x}{L M_{i+1}})^\frac{i}{i-1}\right]\Bigg|_{h_i}^{h_{i+1}}.
        \end{aligned}
    \end{equation*}
    Substitute $h_i$ and $h_{i+1}$ into above equation and rearrange results to obtain, $\forall i \in \{2,\cdots,n-1\}$,
    \begin{align}\label{NE-uniqueness-B1=B2-Bi}
        B_i - B_{i+1} = \frac{L}{i} \Big[ \frac{(M_{i+1})^{i+1}}{(M_{i+2})^i} - \frac{(M_i)^i}{(M_{i+1})^{i-1}} \Big].
    \end{align}
    In particular, for $i = n-1$, we let $M_{i+2} = 1$.
    And for $i = n$, we can obtain $B_n = \int_{h_n}^L xf_2(x)dx = \frac{1}{n} L [1 - (M_n)^n]$.
    Taking summation gives:
    \begin{align*}
        \sum_{i = 2}^{n-1} i (B_i - B_{i+1}) &= \sum_{i = 2}^{n-1} L \Big[ \frac{(M_{i+1})^{i+1}}{(M_{i+2})^i} - \frac{(M_i)^i}{(M_{i+1})^{i-1}} \Big]  \\
        &= L \Big[-\frac{(M_2)^2}{M_3} + \frac{(M_n)^n}{(M_{n+1})^{n-1}}\Big]  \\
        &= L [(M_n)^n]
    \end{align*}
    Let $A = \sum_{i = 2}^{n-1} i (B_i - B_{i+1})$, which is a constant.
    Thus, we have $B_n = \frac{1}{n} \frac{A}{(M_n)^n}[1 - (M_n)^n]$, which further gives 
    \begin{align*}
        M_n = (\frac{A}{n B_n + A})^\frac{1}{n}, \quad L = n B_n + A.
    \end{align*}
    Finally, $(M_i)_{i \in [n]}$ can be obtained by equation \eqref{NE-uniqueness-B1=B2-Bi}, which further gives $(h_i)_{i \in [n]}$ by equation \eqref{eq-hi}, and then  $(F_i)_{i \in [n]}$ by equation \eqref{NE-uniqueness-B1=B2-F2}.
    
    We have obtained unique solution for $L$ and $M_n$, which in turn provide unique solution for $M_i$, $h_i$ and $F_i$ for every $i$. Therefore, the Nash equilibrium is unique.
\end{proof}

\end{document}